\documentclass[11pt, oneside]{article}   	
\usepackage{geometry}                		
\usepackage{graphicx}				

\usepackage{import}
\usepackage{authblk}						
\usepackage{cite}
\usepackage{array}								
\usepackage{amssymb}
\usepackage{amsmath}
\usepackage{mathtools}
\usepackage{braket}
\usepackage{amsthm}
\usepackage{dsfont}							
\usepackage{verbatim}						
\usepackage{MnSymbol}						
\usepackage{mathrsfs}						
\usepackage{mdframed}
\usepackage{caption}
\usepackage[dvipsnames]{xcolor}

\usepackage{float}
\usepackage{placeins}

\RequirePackage{doi}
\usepackage{hyperref}
\hypersetup{
  colorlinks   = true,    
  urlcolor     = blue,    
  linkcolor    = blue,    
  citecolor    = green      
}

\DeclareCaptionType{Protocol}
\DeclareCaptionType{SimBox}[Box]

\usepackage{latexdef}

\usepackage{stmaryrd} 

\newif\ifcomments
\commentsfalse

\newcommand{\puretomixed}[1]{\llbracket #1 \rrbracket}
\DeclareMathOperator*{\Ex}{\mathbb{E}}

\allowdisplaybreaks

\title{Security proof for parallel DIQKD}
\author{Ashutosh Marwah\footnote{email: ashutosh.marwah@outlook.com} \ and Fr\'ed\'eric Dupuis}
\affil{\small{D\'epartement d'informatique et de recherche op\'erationnelle,\\ Universit\'e de Montr\'eal,\\ Montr\'eal QC, Canada}}

\date{\today}							

\begin{document}
\maketitle

\begin{abstract}
	We present a parallel device independent quantum key distribution (DIQKD) protocol based on the CHSH game and prove its security. Using techniques developed for analysing the parallel repetition of anchored non-local games \cite{Bavarian21}, we show that the answers on a small random linear subset of the games in the DIQKD protocol can be simulated as the output of a single-round strategy for playing the CHSH game. Then, we use the unstructured approximate entropy accumulation theorem proven in our companion paper \cite{Marwah24_uni_ch} to establish the smooth min-entropy lower bound required for the security proof. Our approach yields a more information-theoretic and general proof for parallel DIQKD compared to previous proofs.
\end{abstract}

\tableofcontents

\section{Introduction}

Key distribution is a primitive in cryptography, which allows two parties, commonly referred to as Alice and Bob, to establish a secret key for secure communication in the presence of an adversary, Eve. While classical key distribution protocols cannot achieve unconditional security, it is possible to use properties of quantum particles like the no-cloning principle to create provably secure \emph{quantum key distribution (QKD)} protocols \cite{Bennett84, Bennet92, Shor00}. The security of these theoretical protocols can be mathematically proven. However, practical implementations of QKD protocols are vulnerable to side-channel attacks, where Eve exploits imperfections in the preparation and detection devices to extract additional information \cite{Lydersen10, Sajeed15}. For example, Eve could send a laser signal towards Alice's preparation equipment and analyse its reflection to determine the preparation bases, potentially compromising the protocol's security. These attacks stem from the impossibility of completely and accurately modelling Alice and Bob's devices. To circumvent such problems, \emph{device-independent QKD (DIQKD)} \cite{Ekert91,Mayers04,Barrett05, Pironio09} has been developed. DIQKD protocols allow Alice and Bob to use 'untrusted' devices (more precisely, devices that have been prepared by the eavesdropper, Eve, which cannot transmit signals to her once installed in Alice and Bob's labs) to implement key distribution. The security of such protocols relies solely on the properties of quantum mechanics.\\

Typically, in such protocols, Alice and Bob play multiple non-local games using some shared quantum state. A non-local game is defined as:
\begin{definition}[Non-local game]
    A non-local game $G = (\mathcal{X}, \mathcal{Y}, \mathcal{A}, \mathcal{B}, \Pi_{XY}, V)$ is a game played between two cooperating parties called Alice and Bob, who are not allowed to communicate while playing the game. Questions $x \in \mathcal{X}$ and $y \in \mathcal{Y}$ are sampled according to the distribution $\Pi_{XY}$ and distributed to Alice and Bob respectively. Alice and Bob reply with answers $a \in \mathcal{A}$ and $b \in \mathcal{B}$ using some predetermined strategy. Their strategy can be classical, in which case their answers depend only on their questions and some shared randomness, or it can be quantum, in which case they measure some pre-shared quantum state to determine their answers. Alice and Bob win the game if they satisfy the predicate $V(x, y, a, b)$. The maximum winning probability for classical strategies is denoted using $\omega_c(G)$ and for quantum strategies using $\omega_q(G)$.
\end{definition}
One of the most celebrated results in quantum information is the existence of non-local games which one can win with a higher probability using quantum mechanics than with classical strategies \cite{CHSH69}. Moreover, if two parties sharing a quantum state are able to win such a game with a probability exceeding the classical winning probability, then one can place bounds on the amount of information any third party has about their answers to the game, irrespective of the imperfections in the measurement devices used during the game \cite{Masanes06,Miller17,Brown21}. This fact is used to ensure the security of DIQKD protocols.\\

During a DIQKD protocol, Alice and Bob treat their measurement devices and quantum states as black boxes that produce outputs corresponding to the inputs of the non-local game. These protocols involve multiple rounds of a non-local game, which can be played either sequentially or in parallel. Sequential protocols are easier to analyse as they can be broken down into smaller steps, each depending only on the preceding steps. In contrast, during parallel DIQKD, the two parties input all the questions for the multiple games into their devices and receive all the answers at once. This simultaneous nature makes the analysis of these protocols significantly more challenging. However, parallel DIQKD protocols could potentially lead to faster implementations. From a foundational perspective, these protocols remove the sequential or time-ordering assumption in secure key distribution. Thus, investigating the security of these protocols is an interesting question.\\

The challenge of breaking down parallel protocols into smaller, more manageable steps for analysis is also encountered when studying the parallel repetition of non-local games. A parallel repetition of a non-local game $G$, denoted $G^n$, consists of $n$ instances of the game $G$ played simultaneously. Formally, it is defined as:
\begin{definition}[Parallel repetition of a non-local game] 
    The $n$-parallel repetition of a non-local game, $G = (\mathcal{X}, \mathcal{Y}, \mathcal{A}, \mathcal{B}, \Pi_{XY}, V)$, is the non-local game $G^n = (\mathcal{X}^n, \mathcal{Y}^n, \mathcal{A}^n, \mathcal{B}^n, \Pi_{XY}^n, V^n)$, where Alice and Bob are given the questions $x_1^n \in \mathcal{X}^n $ and $y_1^n \in \mathcal{Y}^n$ respectively sampled according to the distribution $\Pi_{XY}^n$ (independently and identically according to $\Pi_{XY}$). Alice and Bob reply with answers $a_1^n \in \mathcal{A}^n$ and $b_1^n \in \mathcal{B}^n$ according to their (classical or quantum) strategy. They win the game if for every $1 \leq i \leq [n]$, they satisfy the predicate $V(x_i, y_i, a_i, b_i)$. 
    \label{def:Par_rep_games}
\end{definition}
For the parallelly repeated game $G^n$, one can always play the optimal strategy for game $G$ independently $n$ times. This yields a lower bound on the winning probability: $\omega_S(G)^n \leq \omega_S(G^n)$ (where strategy $S$ can be either classical or quantum). However, it's natural to ask whether it's possible to significantly outperform this strategy for $G^n$. While it may not be immediately apparent, there are examples of games where $\omega_S(G)^n < \omega_S(G^n)$. For instance, the FFL game \cite[Appendix A]{Holenstein09} exhibits $\omega_c(G^2) = \omega_c(G)<1$ and $\omega_q(G^2) = \omega_q(G)<1$. Roughly speaking, this improvement arises because players can correlate their answers across parallel games, thereby correlating the winning conditions and achieving a higher overall winning probability than independent play would allow. The parallel repetition question asks whether the winning probability of $G^n$ decays exponentially in $n$. For classical strategies, this exponential decay was proven in \cite{Raz98,Holenstein09}. In the quantum case, it has been demonstrated for large classes of games \cite{Cleve08, Jain14, Chailloux14, Chung15, Dinur14, Bavarian21}, but remains an open question for general games. Currently, for general quantum games, the best known bound for $\omega_q(G^n)$ decays only polynomially in $n$ \cite{Yuen16}.\\

The fundamental idea behind these works on parallel repetition is that one can simulate the probability distributions and states created by the strategy for the parallelly repeated game $G^n$, conditioned on an event $\Omega$ defined in terms of a small subset ($\leq \delta n$ for small $\delta > 0$) of the questions and answers, using a single-round strategy for the game $G$. In this single-round simulation strategy, the game $G$ is actually embedded in a particular round $j$ of the game $G^n$. Specifically, it is shown that the distribution of questions at index $j$, conditioned on the event $\Omega$, has the same distribution as the questions of the game $G$. Furthermore, it is demonstrated that one can define appropriate measurements which produce answers $A_j$ and $B_j$ for Alice and Bob, respectively, with the same distribution as they would have in $G^n$ conditioned on $\Omega$. Using this simulation, it can then be shown that the winning probability for round $j$ cannot be significantly larger than that of winning the single-round game $G$. \\

In this work, we apply techniques developed for the analysis of the parallel repetition of anchored games \cite{Bavarian21} to create a proof for parallel DIQKD using CHSH games. Specifically, we demonstrate that for a random subset of size $\delta n$ of the games, Alice's answer for every game in the subset can be approximately viewed as the output of a single-round strategy, similar to the parallel repetition setting. With this result, we can leverage the fact that Alice's answers for a single-round CHSH game are random with respect to the adversary when the CHSH game is won with high enough probability. To prove security of the DIQKD scheme, we need to prove a linear lower bound for the smooth min-entropy of Alice's answers with respect to the adversary's information. We first present our ideas by proving a linear lower bound for the corresponding von Neumann entropy in Section \ref{sec:vN_proof}. This allows us to separate the inherent complexity of one-shot entropies from the parallel repetition based techniques required for the security proof. Following this, we formalise the security proof in the one-shot setting in Section \ref{sec:one_shot_pf}. Showing that the above ideas leads to an accumulation of smooth min-entropy, that is, entropy proportional to the size of the random subset of Alice's answers requires a new entropic tool, called the unstructured entropy accumulation theorem, which we developed in our companion work \cite{Marwah24_uni_ch}. This theorem is a counterpart of the entropy accumulation theorem \cite{Dupuis20} used to prove the security of sequential DIQKD, adapted to work in the parallel setting with approximate conditions.

\subsection{Comparison with prior work}

\cite{Jain20} provided the first proof for parallel DI-QKD. Their QKD protocol is based on the Magic Square game. The security proof for this protocol relies on using the parallel repetition theorem for free games (games where questions have a product distribution) with multiple (more than 2) players. \cite{Jain20} views the setting of DIQKD with Alice, Bob, and Eve as the parallel repetition of a multiplayer game. In this context, $\exp\rndBrk{-H^\epsilon_{\min}(\text{Raw Key} | \text{Eve's information})}$, where the raw key consists of Alice's answers on a random subset, can be interpreted as the winning probability for this game. Since the winning probability decays exponentially in the number of rounds due to the parallel repetition result, the smooth min-entropy can be bounded by $\Omega(n)$. This proof relies on three key properties of the Magic Square game: (1) it samples questions uniformly, (2) there exists a quantum strategy to win it perfectly, and (3) under this perfect strategy, Alice and Bob receive a perfectly correlated uniform random bit.\\

\cite{Vidick17} significantly simplified the security proof given by \cite{Jain20}. The key idea remains viewing the QKD protocol as a parallel repetition of a 3-player game between Alice, Bob and Eve. This 3-player Magic Square game is won if Alice and Bob win the Magic Square game and if Eve correctly guesses Alice's answer. A technique developed for studying non-local games called ``immunization'' \cite{Kempe08} is used by \cite{Vidick17} to prove that the 3-player game has a winning probability strictly less than 1. The parallel repetition theorem for anchored games \cite{Bavarian21} is then used to show that the winning probability of the repeated game is $2^{-\Omega(n)}$. This is subsequently used to bound Eve's guessing probability of Alice's answers and hence the min-entropy for Alice's answers given Eve's system. \cite{Vidick17}'s proof also utilises the same properties of the Magic Square game as \cite{Jain20}. Building on these works, \cite{Jain22} also proves the security of a similar parallel DIQKD protocol in the presence leakage from Alice and Bob's devices. \\

In comparison to these proofs, our proof is not limited by the properties of the games used for the DIQKD protocol. We demonstrate our protocol using the CHSH game, showing how to convert it into an anchored game suitable for parallel DIQKD-- a technique we believe should be applicable to other games as well. Further, our work offers an alternative security proof for parallel DIQKD, employing a more information-theoretic approach. This method decomposes the large quantum device playing parallel CHSH games into smaller single-round CHSH game playing devices. In contrast, \cite{Jain20} and \cite{Vidick17} reduce the security proof to bounding the winning probability of a parallelly repeated game. We hope that our approach can provide greater insight into the problem and aid in solving open problems like the security of parallel device-independent randomness expansion. From another perspective, techniques from parallel repetition lie at the heart of both our proof and those of \cite{Jain20} and \cite{Vidick17}, highlighting the fundamental importance of these methods in analysing parallel protocols.\\

On the downside, our proof strategy couples the security parameter of the DIQKD protocol to its rate. For a choice of security parameter of $\tilde{O}(\epsilon)$, our approach can only prove security for a rate of $\Omega(\epsilon^{192})$. This is not a limitation of the proofs in \cite{Jain20} and \cite{Vidick17}. It might be possible to break this linkage by further refining the unstructured approximate EAT as mentioned in \cite{Marwah24_uni_ch}, but we leave this for future work. Finally, it is important to note that our protocol and the protocols in \cite{Jain20} and \cite{Vidick17} are intended as proofs of concept. The key rates for these parallel DIQKD protocols are currently too small for practical implementation.

\section{Background}

\begin{sloppypar}
	\subsection{Notation}
	
	For a classical probability distribution $p_{AB}$, the conditional probability distribution $p_{A|B}$ is defined as $p_{A|B}(a|b) := \frac{p_{AB}(a,b)}{p_B(b)}$ when $p_B(b) >0$. For the case when $p_B(b) =0$, we define $p_{A|B}$ to be the uniform distribution for our purposes. The probability distribution $q_B p_{A|B}$ for a distribution $q$ on random variable $B$ is defined as $q_B p_{A|B} (b,a) := q_B(b) p_{A|B}(a|b)$.\\
	
	The term normalised (subnormalised) quantum states is used for positive semidefinite operators with unit trace (trace less than $1$). We denote the set of registers a quantum state describes (equivalently, its Hilbert space) using a subscript. Partial states of a quantum state are simply denoted by restricting the set of registers in the subscript. For example, if $\rho_{AB}$ denotes a quantum state on registers $A$ and $B$, then the partial states on registers $A$ and $B$, will be denoted as $\rho_{A}$ and $\rho_{B}$ respectively. In this paper, we follow \cite{Bavarian21} in using the notation $\puretomixed{x} = \ket{x}\bra{x}$ to represent a classical value $x$. We also denote the density operator for a pure state $\ket{\psi}$ as $\psi$.\\
	
	The term ``channel'' is used for completely positive trace preserving (CPTP) linear maps between two spaces of Hermitian operators. A channel $\cN$ mapping registers $A$ to $B$ will be denoted by $\cN_{A \rightarrow B}$. \\
	
	We use the notation $[n]$ to denote the set $\{1,2, \cdots, n\}$. For $n$ quantum registers $(X_1, X_2, \cdots, X_n)$, the notation $X_i^j$ for $i <j$ refers to the set of registers $(X_i, X_{i+1}, \cdots, X_{j})$. For a register $A$, $|A|$ represents the dimension of the underlying Hilbert space. \\
	
	For two Hermitian operators $X$ and $Y$, the operator inequality $X \geq Y$ is used to denote that $X-Y$ is a positive semidefinite operator and $X>Y$ denotes that $X-Y$ is a strictly positive operator. The notation $X \ll Y$ denotes that the support of operator $X$ is contained in the support of $Y$. The identity operator on register $A$ is denoted using $\Id_A$.\\
	
	\subsection{Information theory}
	The trace norm is defined as $\norm{X}_1 := \tr\big(\rndBrk{X^\dag X}^{\frac{1}{2}}\big)$. The fidelity between two positive operators $P$ and $Q$ is defined as $F(P,Q)= \norm{\sqrt{P}\sqrt{Q}}_1^2$. The generalised fidelity between two subnormalised states $\rho$ and $\sigma$ is defined as 
	\begin{align}
		F_\ast(\rho, \sigma) := \rndBrk{\norm{\sqrt{\rho}\sqrt{\sigma}}_1 + \sqrt{(1- \tr\rho)(1- \tr\sigma)}}^2.
	\end{align}
	The purified distance between two subnormalised states $\rho$ and $\sigma$ is defined as 
	\begin{align}
		P(\rho, \sigma) = \sqrt{1- F_{\ast}(\rho, \sigma)}.
	\end{align}
	Throughout this paper, we use base $2$ for both the functions $\log$ and $\exp$. The von Neumann entropy of a state $\rho_{A}$ is defined as 
	\begin{align}
		H(A)_\rho := -\tr\rndBrk{\rho_A \log \rho_A}.
	\end{align}
	For a state $\rho_{AB}$, the conditional entropy of register $A$ with respect to register $B$ is defined as 
	\begin{align}
		H(A|B)_{\rho} := H(AB) - H(B).
	\end{align}
	For a subnormalised state $\rho_{AB}$, the min-entropy of register $A$ given register $B$ can be defined as 
	\begin{align}
		H_{\min}(A|B)_{\rho} &:= \sup \curlyBrk{\lambda \in \mathbb{R}: \text{ there exists state }\sigma_B \text{ such that } \rho_{AB} \leq 2^{-\lambda} \Id_A \otimes \sigma_B}.
	\end{align} 	
	For the purpose of smoothing, we define the $\epsilon$-ball around a subnormalised state $\rho$ as the set
	\begin{align}
		B_{\epsilon}(\rho) = \{ \tilde{\rho} \geq 0 : P(\rho, \tilde{\rho}) \leq \epsilon \text{ and } \tr\tilde{\rho} \leq 1\}.
	\end{align}
	The smooth min-entropy of register $A$ given $B$ for the state $\rho_{AB}$ is defined as
	\begin{align}
		H_{\min}^{\epsilon}(A|B)_{\rho} := \max_{\tilde{\rho} \in B_{\epsilon}(\rho)} H_{\min}(A|B)_{\tilde{\rho}}.
	\end{align} 
	The amount of randomness, which can be extracted from a register independent of some side information is quantified by the smooth min-entropy as the leftover hashing lemma shows below.
	\begin{lemma}[Leftover hashing lemma \cite{Renner05, Tomamichel10}]
		For a set of 2-universal hash functions $\mathcal{F}$ from $A$ to $Z$ and a classical-quantum state $\rho_{AE} = \sum_a p(a) \ket{a} \bra{a} \otimes \rho_{E}^{(a)}$, let $\rho_{ZEF}$ be the state produced by applying a random hash function $F \in \mathcal{F}$ to register $A$ to produce the output register $Z$. This state satisfies 
		\begin{align}
			\frac{1}{2} \norm{\rho_{ZEF} - \tau_{Z} \otimes \rho_{EF}} \leq 2\epsilon+ 2^{\frac{1}{2}(\log |Z| - H_{\min}^\epsilon (A|E))} 
		\end{align}
		where $\tau_Z := \frac{1}{|Z|} \Id_Z$ is the completely mixed state on register $Z$. 
		\label{lemm:leftover_hashing}
	\end{lemma}
	So, if one chooses a family of hash functions with output length $|Z| = H_{\min}^\epsilon (A|E)_{\rho} - 2\log 1/\epsilon$, the output state $\rho_{ZEF}$ is $3\epsilon$ close to $\tau_Z \otimes \rho_{EF}$. Thus, one can extract $H_{\min}^\epsilon (A|E)_{\rho} - O(1)$ amount of randomness independent of the adversary for the state $\rho$. 

	\subsection{QKD security proof}
	\label{sec:qkd_bg}
	
	Quantum key distribution (QKD) protocols can be broken into two stages. During the first stage Alice and Bob use quantum states and measurements to generate a partially secret raw key. In the second stage, they use classical post-processing to ensure that their keys match and are completely random with respect to an adversary. This post-processing involves using information reconciliation, raw key validation and privacy amplification protocols. These also determine the entropic bounds required for proving the security of QKD. We provide a brief description of these protocols and how they are used in QKD. For more details, we refer the reader to \cite[Chapter 6]{Renner06} and \cite[Section 4.2.2]{Friedman20}. \\

	An information reconciliation (IR) protocol \cite{Brassard94} allows two parties holding correlated strings to derive a common string, while minimising communication. Raw key validation simply uses a random hash function from a 2-universal family to verify that Alice and Bob hold the same raw key. Raw key validation is often included in the information reconciliation protocol itself. We describe it as a separate step here, since it makes correctness evident and aids our explanation. \\

	\cite[Lemma 6.3.4]{Renner06} demonstrates a one-way protocol from Alice to Bob for which the number of bits communicated during information reconciliation and raw key validation, which we will simply refer to as $\text{leak}_{\text{IR}}$, can be bounded as 
	\begin{align}
		\text{leak}_{\text{IR}} \leq H^{\epsilon}_{\max}(X|Y) + O(1)
		\label{eq:IR_cost}
	\end{align}
	where $X$ and $Y$ are Alice and Bob's raw keys respectively, and the parameter $\epsilon \in [0,1]$.  \\

	In QKD protocols, information reconciliation ensures correctness. For example, in a QKD protocol, adversarial interference and noise may cause the raw keys generated by Alice and Bob to be unequal. Through information reconciliation, Bob can adjust his raw key to match Alice's. \\

	Privacy amplification, allows Alice and Bob to extract a random secret key from their shared raw key obtained through information reconciliation. The process involves selecting a random 2-universal hash function and applying it to their raw key. When parameters are appropriately chosen, the Leftover Hashing Lemma  (Lemma~\ref{lemm:leftover_hashing}) guarantees that the resulting key is independent of the adversary's state.\\

	Let's suppose that Alice and Bob's raw keys before applying classical post-processing are $A_1^n$ and $B_1^n$, and Eve's state is $E$. If Alice sends message $C$ to Bob during information reconciliation, then the Leftover Hashing Lemma guarantees that the QKD protocol's key length is lower bounded (up to a constant) by:
	\begin{align}
		H_{\min}^\epsilon (A_1^n|E C)_{\rho} &\geq H_{\min}^\epsilon (A_1^n|E )_{\rho} - \text{leak}_{\text{IR}}
	\end{align}
	using the dimension bound \cite[Lemma 6.8]{TomamichelBook16}. The information reconciliation cost can be bounded as
	\begin{align}
		\text{leak}_{\text{IR}} \leq H^{\epsilon'}_{\max}(A_1^n|B_1^n)_{\rho_{\text{honest}}} + O(1).
	\end{align}
	$H^{\epsilon'}_{\max}$ in the bound above is evaluated on $\rho_{\text{honest}}$, which represents the state produced at the end of the QKD rounds in the absence of an adversary. In the presence of an adversary, this amount of communication may not suffice for successful information reconciliation, and in that case the raw key validation step would fail with high probability (see \cite[Section 6.3]{Renner06} for a detailed discussion), which is permissible according to the security definition.\\

	This term is typically straightforward to bound as the form of $\rho_{\text{honest}}$ is known. It can usually be expressed as $nf(e)$, where $e$ quantifies the protocol noise and $f$ is a small function such that $f(e) \rightarrow 0$ as $e \rightarrow 0$. Therefore, while proving security for QKD, the main challenge lies in establishing a linear lower bound for $H_{\min}^\epsilon (A_1^n|E)_{\rho}$ in cases where the protocol does not abort. \\

	If one is able to prove that 
	\begin{align}
		H_{\min}^\epsilon (A_1^n|E )_{\rho} - H^{\epsilon'}_{\max}(A_1^n|B_1^n)_{\rho_{\text{honest}}} \geq r n -O(1)
	\end{align}
	for some $r >0$, then Alice and Bob can produce a secure key of length $(rn - O(1))$ according to the Leftover Hashing Lemma. The \emph{key rate} for a QKD protocol is defined as its key length divided by the number of rounds. In this case, it asymptotically tends to $r$.  

\end{sloppypar}

\subsection{CHSH game and related properties}

Recall that a non-local game $G$ is represented as $G = (\mathcal{X}, \mathcal{Y}, \mathcal{A}, \mathcal{B}, \Pi_{XY}, V)$, where $\mathcal{X}$ and $\mathcal{Y}$ are the sets of Alice and Bob's questions, $\mathcal{A}$ and $\mathcal{B}$ are the sets of their answers, $\Pi_{XY}$ is the probability distribution of their questions and $V$ is the winning predicate. \\

For the standard CHSH game, we have $\mathcal{X} =\mathcal{Y} =\mathcal{A} =\mathcal{B} = \{0,1\}$, $\Pi_{XY}$ is the uniform distribution on all possible questions and the predicate $V(x,y,a,b) = \neg [a \oplus b \oplus (x \wedge y)]$. We refer to this game as the CHSH game or the 2CHSH game. The maximum winning probabilities for this game are $\omega_c(2\text{CHSH}) = 3/4$ and $\omega_q(2\text{CHSH}) = (2+\sqrt{2})/4$. \\

This nomenclature helps distinguish it from the \emph{3CHSH game}, which is usually used for DIQKD. In this game, Alice's questions lie in $\{ 0,1\}$ and Bob's question lie in $\{ 0,1,2\}$. These questions are sampled according to the probability distribution 
\begin{align*}
	P_{XY} (x,y) = \begin{cases} 1-\nu \quad &\text{if } x=0, y=2 \\ 
	\nu/4 \quad & \text{if } x, y \in \{0,1\}
	\end{cases}
\end{align*}
where $\nu \in (0,1)$ is a parameter which we will fix later. For the questions $x,y \in \{0,1\}$, Alice and Bob win this game if they win the standard CHSH game, that is, if $\neg [a \oplus b \oplus (x \wedge y)]$ is true. On questions $(x,y) = (0,2)$, they win if their answers are equal. The maximum winning probabilities for this game are $\omega_c(3\text{CHSH}) = 1-\frac{\nu}{4}$ and $\omega_q(3\text{CHSH}) = 1 - \frac{(2 -\sqrt{2})\nu}{4}$. \\

We require the following \emph{anchoring transform} for non-local games in order to describe our protocol. 
\begin{definition}[Anchoring transform {\cite{Bavarian21}}]
    Let $G=(\mathcal{X}, \mathcal{Y}, \mathcal{A}, \mathcal{B}, \Pi_{XY}, V)$ be a non-local game and $0<\alpha \leq 1$. In the $\alpha$-anchored game $G_\perp :=(\mathcal{X}\cup \{\perp\}, \mathcal{Y}\cup \{\perp\}, \mathcal{A}, \mathcal{B}, \Pi^\perp_{XY}, V_\perp)$, the Referee first uses $\Pi_{XY}$ to sample questions $x,y$ for Alice and Bob. Then, randomly and independently with probability $\alpha$, he replaces each of $x$ and $y$ with an auxiliary ``anchor'' symbol $\perp$ to obtain the questions for the anchored game $G_\perp$. Alice and Bob win the game if either one of the questions was $\perp$, or if their answers $a,b$ satisfy the original game's predicate, that is, $V(x,y,a,b)=1$. The quantum winning probability for the anchored game satisfies 
    \begin{align*}
        \omega_q(G_\perp) = 1 - (1-\alpha)^2 (1-\omega_q(G)).
    \end{align*}
    \label{def:Anchoring_transform}
\end{definition}
We call the game obtained after applying the anchoring transform for $\alpha \in (0,1)$ to the 3CHSH game, the 3CHSH$_\perp$ game. Once again, we will consider $\alpha$ to be a parameter and fix it later. \\

A strategy $S$ for a non-local game $G = (\mathcal{X}, \mathcal{Y}, \mathcal{A}, \mathcal{B}, \Pi_{XY}, V)$ consists of the tuple $(\Psi_{E_A E_B}, \{ A_x \}_{x \in \mathcal{X}}, \{ B_y\}_{y \in \mathcal{Y}})$ where $\Psi_{E_A E_B}$ is the quantum state shared by Alice and Bob, $A_x$ is the measurement used by Alice on question $x$, and $B_y$ is the measurement used by Bob on question $y$.\\

If a quantum strategy wins the 2CHSH game with a probability strictly greater than $3/4$, then Alice's answer is guaranteed to be random with respect to any purification of the initial state held by the adversary, Eve, and the questions for the game. This statement was first proven in \cite{Pironio09} (also stated here in Lemma \ref{lemm:2CHSH_entropy}). We state a counterpart for this lemma for the 3CHSH$_\perp$ game. It follows fairly easily from the bound in \cite{Pironio09}. We prove it in Appendix \ref{sec:app_single_box}. 

\begin{lemma}
   Suppose that a given quantum strategy for the 3CHSH$_{\perp}$ game starting with $\rho_{E_A E_B E}^{(0)}$ wins the 3CHSH$_{\perp}$ game with probability $\omega \in \lr{[ 1- \frac{(1-\alpha)^2 \nu}{4} , 1- \frac{2-\sqrt{2}}{4} (1-\alpha)^2 \nu }]$. Let $X$ and $Y$ be Alice and Bob's questions during the game, and $A$ and $B$ be their answers produced according to this strategy. Then, for the post measurement state $\rho_{XYABE}$, we have
   \begin{align}
		H(AB| E XY)_\rho \geq H(A| E X)_\rho \geq (1- \alpha)F(g_{\alpha, \nu}(\omega))
		\label{eq:3CHSH_anch_entropy}
   \end{align}
	where the functions $F$ and $g_{\alpha, \nu}(\omega)$ are given by
	\begin{align}
		F(x) = 1 - h \rndBrk{\frac{1}{2} + \frac{1}{2} \sqrt{3 - 16\, x\left(1 - x\right)}} \quad \text{for}\quad g_{\alpha, \nu}(\omega) = 1 - \frac{1 - \omega}{\nu (1 - \alpha)^2}.
		\label{eq:3CHSH_bd_phi}
	\end{align}
   \label{lemm:3CHSH_anch_entropy}
\end{lemma}

\section{Protocol}

Before we describe our protocol for parallel DIQKD, we must introduce a key result from \cite{Bavarian21}. For every $\alpha$-anchored game $G_{\perp} := (\mathcal{X}, \mathcal{Y}, \mathcal{A}, \mathcal{B}, P_{XY},V)$ (a game produced by applying the anchoring transform), \cite[Section 4.1]{Bavarian21} shows that the probability distribution $P_{XY}$ can be extended to the distribution $\hat{P}_{\Omega X Y}$ such that 
\begin{align}
	&\hat{P}_{XY} := P_{XY} \\
	&\hat{P}_{\Omega X Y} = \hat{P}_{\Omega} \hat{P}_{X|\Omega} \hat{P}_{Y|\Omega}.
\end{align}
This extension plays a crucial role in our protocol, as we will demonstrate in the subsequent sections. In this work, we will call the random variable $\Omega$ the \emph{seed randomness} for the questions. We note that $\Omega$ is defined such that it can be sampled efficiently by Alice. Given $\Omega$, Alice and Bob can independently sample their questions for the game, $G_\perp$. \\

For the rest of the paper, let the tuple $(\mathcal{X}, \mathcal{Y}, \mathcal{A}, \mathcal{B}, P_{XY},V)$ represent the 3CHSH$_\perp$ game. We will drop the hat notation while referring to the extension distribution for this question distribution. We simply refer to it as $P_{\Omega X Y}$.\\

We present the protocol for parallel DIQKD in Protocol \ref{prot:par_diqkd_prot}. Note that this protocol does not fully reveal the questions to Eve. Unlike sequential DI-QKD, where Alice and Bob can reveal all of their questions to Eve, the security proofs for parallel DI-QKD require that Alice and Bob only reveal a small fraction of their questions to Eve\footnote{In the parallel DIQKD protocol of \cite{Jain20} only a small fraction of the questions are announced publicly. In the protocol used by \cite{Vidick17}, Alice and Bob can reveal any fraction smaller than 1.}. Similarly, in our setting, since only $\Omega_1^n$ are revealed to Eve, only a fraction of the information about the questions is leaked to Eve. Importantly, this is the main obstacle to extending these techniques to prove security for parallel device-independent randomness extraction. It should be noted that the previous protocols and proofs \cite{Jain20, Vidick17} required the probability distributions of Alice and Bob's questions to be a product distribution, so that both Alice and Bob could sample their questions independently. However, by utilising the anchoring transform and the seed randomness we are able to relax this constraint. 

\begin{figure}
  \begin{mdframed}
	\textbf{Parameters:}
	\begin{itemize}
		\renewcommand{\labelitemi}{--}
		\item $\alpha, \nu \in \rndBrk{0, 0.1}$ are parameters for the 3CHSH$_\perp$
		\item $\delta \in \rndBrk{0,\frac{1}{2}}$ determines the size of the raw key
		\item $\omega_{\text{th}} \in \rndBrk{ 1- \frac{(1-\alpha)^2 \nu}{4} , 1- \frac{2-\sqrt{2}}{4} (1-\alpha)^2 \nu }$ is the threshold for the winning probability on the test round
		\item $\gamma \in (0,1)$ parameter for sampling testing rounds.
	\end{itemize}
  \textbf{Parallel DIQKD protocol}
  \begin{enumerate}
    \item Alice randomly samples $\Omega_1^n$ independently and identically. 
    \item Alice sends $\Omega_1^n$ to Bob. 
    \item Alice and Bob use $\Omega_1^n$ to sample the questions $X_1^n$ and $Y_1^n$ for the 3CHSH$_\perp$ game.
    \item Alice and Bob use their questions to play $n$ 3CHSH$_\perp$ games. Let $A_1^n$ and $B_1^n$ be the answers. 
    \item Alice randomly selects a subset $J = \{ I_1, I_2, \cdots, I_t \} $ of size $t = \frac{\delta}{\log |\mathcal{A}||\mathcal{B}| + \delta} n$. She announces this subset to Bob.
    \item For each $i \in [t]$, Alice randomly selects a $T_j \in \{0,1\}$ with probability $\Pr(T_j = 1) = \gamma$. Let $S := \curlyBrk{I_j : j \in [t] \text{ and } T_j =1} \subseteq J$. She announces $S$, her questions $X_S$, and answers $A_S$ for this subset of the games. 
    \item Bob checks whether $\sum_{i \in S} V(X_i, Y_i, A_i, B_i) \geq \gamma \omega_{\text{th}} t$. Alice and Bob abort if this is not satisifed. 
    \item $A_J$ and $B_J$ are Alice and Bob's raw keys. They use information reconciliation, raw key validation and privacy amplification to create a secret key. 
  \end{enumerate}
  \end{mdframed}
  {\captionof{Protocol}{}
  \label{prot:par_diqkd_prot}}
\end{figure}

\section{Setup}

As indicated in Protocol \ref{prot:par_diqkd_prot} we let $\Omega_1^n$ be the randomness seed for the questions shared by Alice. $X_1^n$ and $Y_1^n$ denote Alice and Bob's questions for the $n$ 3CHSH$_\perp$ games during the protocol, and $A_1^n$ and $B_1^n$ denote their answers for these games. \\

In order to analyse the protocol, we fix a strategy for Eve. Let's suppose that Eve distributes the registers $E_A$ and $E_B$ of the pure state $\psi_{E_A E_B E}$ between Alice and Bob, keeping register $E$ for herself\footnote{The state can be considered pure. If it were not, then purifying it and providing Eve the purification register would only increase Eve's information.}. Let Alice and Bob use measurements $\{A_{x_1^n}^{E_A}(a_1^n)\}_{a_1^n}$ and $\{B_{y_1^n}^{E_B}(b_1^n)\}_{b_1^n}$ to measure their registers $E_A$ and $E_B$ respectively given questions $x_1^n$ and $y_1^n$. \\

The state after all the 3CHSH$_\perp$ games have been played will be denoted as $\rho$. We have that 
\begin{align}
	\rho_{\Omega_1^n X_1^n Y_1^n A_1^n B_1^n E} &= \sum_{\omega_1^n, x_1^n, y_1^n} P_{\Omega_1^n X_1^n Y_1^n}(\omega_1^n, x_1^n, y_1^n) \puretomixed{\omega_1^n, x_1^n, y_1^n} \nonumber \\ 
	&\qquad \qquad \otimes \sum_{a_1^n, b_1^n} \puretomixed{a_1^n, b_1^n} \otimes \tr_{E_A E_B} \rndBrk{A_{x_1^n}^{E_A}(a_1^n) \otimes B_{y_1^n}^{E_B}(b_1^n) \psi_{E_A E_B E}}
	\label{eq:rho_form1}
\end{align}
where $P_{\Omega_1^n X_1^n Y_1^n}$ is the i.i.d distribution $P_{\Omega X Y}^{\otimes n}$. We can also write the above as 
\begin{align}
	\rho_{\Omega_1^n X_1^n Y_1^n A_1^n B_1^n E} &= \sum_{\omega_1^n, x_1^n, y_1^n} P_{\Omega_1^n X_1^n Y_1^n}(\omega_1^n, x_1^n, y_1^n) \puretomixed{\omega_1^n, x_1^n, y_1^n} \nonumber \\
	&\qquad \qquad \otimes \sum_{a_1^n, b_1^n} P_{A_1^n B_1^n | X_1^n Y_1^n}(a_1^n, b_1^n | x_1^n, y_1^n) \puretomixed{a_1^n, b_1^n} \otimes \rho_{E}^{(x_1^n, y_1^n, a_1^n, b_1^n)}
	\label{eq:rho_form2}
\end{align}
for $P_{A_1^n B_1^n | X_1^n Y_1^n}(a_1^n, b_1^n | x_1^n, y_1^n) = \tr\rndBrk{A_{x_1^n}^{E_A}(a_1^n) \otimes B_{y_1^n}^{E_B}(b_1^n) \psi_{E_A E_B E}}$ and 
\begin{align}
	\rho_{E}^{(x_1^n, y_1^n, a_1^n, b_1^n)} := \frac{\tr_{E_A E_B} \rndBrk{A_{x_1^n}^{E_A}(a_1^n) \otimes B_{y_1^n}^{E_B}(b_1^n) \psi_{E_A E_B E}}}{P_{A_1^n B_1^n | X_1^n Y_1^n}(a_1^n, b_1^n | x_1^n, y_1^n)}.
\end{align}
We will also use the notation $P_{\Omega_1^n X_1^n Y_1^n A_1^n B_1^n} := \rho_{\Omega_1^n X_1^n Y_1^n A_1^n B_1^n}$. \\

Using $\rho$ in the form in Eq. \ref{eq:rho_form2}, we can further define the state of register $E$ conditioned on other classical variables, for example register $E$ conditioned on $\omega_1^n$ is
\begin{align}
	\rho_{E}^{(\omega_1^n)} &= \sum_{x_1^n, y_1^n, a_1^n, b_1^n} P_{X_1^n Y_1^n A_1^n B_1^n| \Omega_1^n}(x_1^n, y_1^n, a_1^n, b_1^n | \omega_1^n) \rho_{E}^{(x_1^n, y_1^n, a_1^n, b_1^n)} \\
	&= \Ex_{x_1^n y_1^n a_1^n b_1^n| \omega_1^n} \sqBrk{\rho_{E}^{(x_1^n, y_1^n, a_1^n, b_1^n)}}.
\end{align}
Finally, we let $F$ (for ``fail'') denote the event that the protocol aborts. 

\section{Key results from \cite{Bavarian21}}

In Protocol \ref{prot:par_diqkd_prot}, we have chosen 
\begin{align}
  t = \frac{\delta}{\log |\mathcal{A}||\mathcal{B}| + \delta} n.
\end{align}
In our security analysis, we will fix a subset $C \subseteq J$ and show that for any such set we can embed a single-round 3CHSH$_\perp$ game in a random index outside this set. The value of $t$ above guarantees that for any such subset $C \subseteq J$, we have 
\begin{align}
  \frac{|C|}{n-|C|} \log |\mathcal{A}| |\mathcal{B}| \leq \delta.
\end{align}
The results established in \cite{Bavarian21} for parallel repetition settings remain applicable in our context, with minor modifications (see Appendix \ref{sec:anch_games_result_carry} for a detailed discussion). Specifically, these results can be adapted to our setting by introducing a reference register $E$. In this section, we present the key definitions and results from \cite{Bavarian21} which are used in our security proof.\\

Following \cite[Section 4.2]{Bavarian21}, for the subset $C \subseteq [n]$ and $i \in C^c= [n]\setminus C$, we define the \emph{dependency breaking variable} $R_{-i}$ as
\begin{align}
  R_{-i} := (\Omega_j)_{j \in [n] \setminus (C \cup \{i\})}\cup (X_C, Y_C, A_C, B_C). 
  \label{eq:R_i_defn}
\end{align}
We also use the shorthand $\Omega_{-i}$ for $(\Omega_j)_{j \in [n] \setminus (C \cup \{i\})} \cup (X_C, Y_C)$. \\

For questions $x_1^n$ and $y_1^n$, we define the following measurements for answers on the subset $C \subseteq [n]$:
\begin{align}
	A_{x_1^n}^{E_A}(a_C) &= \sum_{a_1^n | a_C} A_{x_1^n}^{E_A}(a_1^n) \\
	B_{y_1^n}^{E_B}(b_C) &= \sum_{b_1^n | b_C} B_{y_1^n}^{E_B}(b_1^n).
\end{align}
Here $a_1^n | a_C$ denotes strings $a_1^n$ consistent with $a_C$. $b_1^n | b_C$ is defined similarly. \\

For subset $C \subseteq [n]$, $i \in C^c$, the seed $\omega_{-i} = (\omega_j)_{j \in [n] \setminus (C \cup \{i\})} \cup (x_C, y_C)$ (an instantiation of $\Omega_{-i}$), and questions $x$ and $y$ we also define the measurements:
\begin{align}
	A_{\omega_{-i}, x}^{E_A}(a_1^n) :=& \Ex_{X_1^n | \Omega_{-i} = \omega_{-i}, X_i = x} A_{x_1^n}^{E_A}(a_1^n) \\
	=& \sum_{x_1^n} P_{X_1^n | \Omega_{-i} X_i }(x_1^n | \omega_{-i}, x) A_{x_1^n}^{E_A}(a_1^n)
\end{align} 
and
\begin{align}
	B_{\omega_{-i}, y}^{E_B}(b_1^n) :=& \Ex_{Y_1^n | \Omega_{-i} = \omega_{-i}, Y_i = y} B_{y_1^n}^{E_B}(b_1^n) \\
	=& \sum_{y_1^n} P_{Y_1^n | \Omega_{-i} Y_i }(y_1^n | \omega_{-i}, y) B_{y_1^n}^{E_B}(b_1^n).
\end{align}
Also define
\begin{align}
	A_{\omega_{-i}, x}^{E_A}(a_C) &:= \sum_{a_1^n | a_C} A_{\omega_{-i}, x}^{E_A}(a_1^n) \\
	B_{\omega_{-i}, y}^{E_B}(b_C) &:= \sum_{b_1^n | b_C} B_{\omega_{-i}, y}^{E_B}(b_1^n).
\end{align}
For subset $C \subseteq [n]$, $i \in C^c$, $r_{-i} = (\omega_{-i}, a_C, b_C)$ (an instance of $R_{-i}$) and questions $X_i = x$ and $Y_i = y$, we also define the unnormalised state
\begin{align}
	& \ket{\Phi^{(r_{-i}, x, y)}}_{E_A E_B E} := \sqrt{A_{\omega_{-i}, x}^{E_A}(a_C)} \otimes \sqrt{B_{\omega_{-i}, y}^{E_B}(b_C)} \ket{\psi}_{E_A E_B E}
\end{align}
and its normalisation 
\begin{align}
	\ket{\tilde{\Phi}^{(r_{-i}, x, y)}}_{E_A E_B E} := \frac{1}{\norm{\ket{\Phi^{(r_{-i}, x, y)}}}} \ket{\Phi^{(r_{-i}, x, y)}}_{E_A E_B E}.
\end{align}
It is shown in \cite[Proposition 4.9]{Bavarian21} that 
\begin{align}
	\norm{\ket{\Phi^{(r_{-i}, x, y)}}} = \rndBrk{P_{A_C B_C | \Omega_{-i} X_i Y_i} (a_c, b_c | \omega_{-i}, x ,y)}^{1/2}.
\end{align}

The following result from \cite{Bavarian21} is the key towards showing that it is possible to simulate the answers produced by the parallelly repeated strategy at a random index outside the set $C$, using a strategy for the single-round 3CHSH$_\perp$ game which embeds the game at this index.

\begin{proposition}[{\cite[Proposition 5.1]{Bavarian21}}]
	For every $C \subseteq J$, $i \in C^c$, dependency breaking variable $r_{-i}$, and questions $x$ and $y$, there exist unitaries $U_{r_{-i}, x}^{E_A}$ acting on $E_A$ and $V_{r_{-i}, y}^{E_B}$ acting on $E_B$ such that 
	\begin{align}
		\Ex_{I}\Ex_{R_{-i}} \Ex_{XY} \norm{U_{r_{-i}, x}^{E_A} \otimes V_{r_{-i}, y}^{E_B} \otimes \Id^E \ket{\tilde{\Phi}^{(r_{-i}, \perp, \perp)}}_{E_A E_B E} - \ket{\tilde{\Phi}^{(r_{-i}, x, y)}}_{E_A E_B E}} = O(\delta^{1/16}/\alpha^3),
	\end{align}
	where $\Ex_{I}$ denotes expectation over index $I$ which is sampled uniformly at random from $C^c$, $\Ex_{R_{-i}}$ denotes expectation over $r_{-i}$ sampled from $P_{R_{-i}}$ and $\Ex_{XY}$ denotes expectation over the questions sampled according to the question distribution for the single-round game $P_{XY}$.
	\label{prop:unit_steering}
\end{proposition}

\section{Proof approach with von Neumann entropies}
\label{sec:vN_proof}
\begin{sloppypar}

We will first demonstrate our proof strategy by proving the requisite entropic statements with von Neumann entropies instead of one-shot entropies. This approach allows us to separate the inherent complexity of one-shot entropies from the core concepts of the proof. We can focus on using the results from the parallel repetition setting for analysing DIQKD and establishing a clear roadmap for the security proof, instead of grappling with the various technicalities associated with one-shot entropies. In the next section, we will proceed to develop a comprehensive one-shot security proof.\\

For simplicity, we also set aside the testing procedure and the conditioning of the state on the protocol not aborting. Instead, we simply assume that the state $\psi_{E_A E_B E}$ and the measurements set up by Eve are such that 
\begin{align}
	\Pr_{\rho} \sqBrk{\frac{1}{n}\sum_{i \in [n]} W_i  \geq \omega_{\text{th}} } \geq 1- \epsilon
	\label{eq:HighAvgWinningProb}
\end{align}
where $\epsilon$ is negligibly small. This assumption is effectively equivalent to stating that the protocol only aborts with a negligible probability. Once again, this restriction allows us to concentrate on the more challenging aspects of the security proof. In the next section, we will see that once we properly account for the testing procedure, this assumption can be eliminated.\\

As we discussed in Sec. \ref{sec:qkd_bg}, in order to prove that a protocol can securely produce a key of length $\Omega(n)$ (that is, it has a positive key rate), we need to show that the difference
\begin{align}
	H^{\epsilon}_{\min} (A_J | E \Omega_1^n S J A_S X_S )_{\rho_{|\lnot F}} - H_{\max}^{\epsilon'}(A_J | B_J J)_{\rho_{\text{honest}}} \geq \Omega(n)
\end{align}
for small $\epsilon$ and $\epsilon'$. $\rho_{\text{honest}}$ above is the state produced at the end of a protocol with no adversary, Eve. \\

We use the conditional entropy $H(A_J | E \Omega_1^n J)_{\rho}$ as proxy for the entropy $H^{\epsilon}_{\min} (A_J | E \Omega_1^n S J A_S X_S )_{\rho_{|\lnot F}}$, which quantifies the amount of randomness that can be safely extracted from the Alice's raw key using privacy amplification. Similarly, the entropy $H(A_J | B_J J)_{\rho_{\text{honest}}}$ serves as a proxy for the information reconciliation cost, which is given by $H^{\epsilon'}_{\max}(A_J | B_J J)_{\rho_{\text{honest}}}$. In our von Neumann security proof, we aim to demonstrate that 
\begin{align}
	H(A_J | E \Omega_1^n J)_{\rho} - H(B_J | A_J J)_{\rho_{\text{honest}}} \geq \Omega(n).
\end{align}

\end{sloppypar}
\subsection{Bounding entropy production for privacy amplification}

We begin by showing that Eve's uncertainty of Alice's answers $A_J$ measured using von Neumann entropy is high, that is, 
\begin{align}
    H(A_J | E \Omega_1^n J)_{\rho} \geq \Omega(t).
\end{align}
This entropy can be expanded using the chain rule as 
\begin{align}
    H(A_J | E \Omega_1^n J)_{\rho} &= \sum_{k=1}^{t} H(A_{I_k} | E \Omega_1^n A_{I_1} \cdots A_{I_{k-1}} J)_{\rho}.
    \label{eq:chain_rule_for_privacy_amplification}
\end{align}
Further, we can write the term inside the summation above as the expectation
\begin{align}
	H(A_{I_k} | E \Omega_1^n A_{I_1} \cdots A_{I_{k-1}} J)_{\rho} &= \Ex_{i_1^{k-1}} \left[ H(A_{I_k} | E \Omega_1^n A_{i_1} \cdots A_{i_{k-1}} I_k)_{\rho} \right].
\end{align}
For the rest of this section, we focus our attention on this term. Let us fix the choice of random variables $I_{1}^{k-1} = i_{1}^{k-1}$. We then define the subset $C := \{i_{1}, i_2, \ldots, i_{k-1}\} \subseteq J$. This notation allows us to maintain clarity and enables us to collectively bound the term inside the expectation for various values of $k$ and $i_{1}^{k-1}$. \\

The term $H(A_{I_k} | E \Omega_1^n A_{i_1} \cdots A_{i_{k-1}} I_k)_{\rho}$ will be bounded in two steps. In the first step, we lower bound it using $H(A_{I_k} | E I_k R_{-I_k} X_{I_k} )_{\rho}$, where $R_{-I_k}$ is the dependency breaking random variable defined with respect to $C$ in Eq. \ref{eq:R_i_defn}. In the second step, we show that it is possible to approximately simulate the state $\rho_{A_{I_k} I_k R_{-I_k} X_{I_k} E}$ using a quantum strategy for a single instance of the game, which has high winning probability. This allows us to use the single-round entropy bound for the 3CHSH$_\perp$ game (Lemma \ref{lemm:3CHSH_anch_entropy}) to lower bound the entropy of the answer $A_{I_k}$ for the simulated state, and subsequently for $\rho$ as well.\\

\noindent \textbf{\normalsize Step 1: Reduction to parallel repetition variables}\\

\noindent In the first step, we will show that 
\begin{align}
    H(A_{I_k} | E \Omega_1^n A_{i_1} \cdots A_{i_{k-1}} I_j)_{\rho} \geq H(A_{I_k} | E I_k R_{-I_k} X_{I_k} )_{\rho}
\end{align} 
where $R_{-I_k} = (\Omega_{-I_k}, X_C, Y_C, A_C, B_C)$. This is fairly simple. It only requires one to use Markov chain properties and data processing. Informally speaking, Alice and Bob's quantum devices only ever get to \emph{see} the questions $X_1^n$, hence, the uncertainty of the answers only decreases if we change some $\Omega_i$s with $X_i$s. We formalise this argument in the following lemma. 

\begin{lemma}
	Let $\rho^{(0)}_{\Omega_1^2 X_1^2 E_A E} = \rho_{\Omega_1^2 X_1^2} \otimes \rho^{(0)}_{E_A E}$ be a classical-quantum density operator, where $\Omega_1^2$ and $X_1^2$ are classical registers and $E_A$ and $E$ are quantum registers. Further, suppose that $\Omega_1 X_1$ and $\Omega_2 X_2$ are sampled independently, that is, $\rho_{\Omega_1^2 X_1^2} = \rho_{\Omega_1 X_1} \rho_{\Omega_2 X_2}$. If the register $E_A$ is measured according to measurement operators $\{A_{x_1^2}^{E_A}(a_1^2)\}_{a_1^2}$, which depend only on the classical register $X_1^2$, then $\Omega_1 \leftrightarrow X_1 \leftrightarrow \Omega_2 A_1^2 E$ forms a Markov chain. 
	\label{lemm:MarkovChainOmegaXA}
\end{lemma}
\begin{proof}
	We can write $\rho^{(0)}_{\Omega_1^2 X_1^2 E_A E}$ as
	\begin{align*}
		\rho^{(0)}_{\Omega_1^2 X_1^2 E_A E} = \sum_{\omega_1^2 x_1^2} \rho(x^2_1) \rho(\omega_1|x_1)\rho(\omega_2|x_2) \puretomixed{\omega_1^2 x_1^2} \otimes \rho^{(0)}_{E_A E}.
	\end{align*}
	The post-measurement state is 
	\begin{align*}
		\rho_{\Omega_1^2 X_1^2 A_1^2 E} = \sum_{\omega_1^2 x_1^2} \rho(x^2_1) \rho(\omega_1|x_1)\rho(\omega_2|x_2) \puretomixed{\omega_1^2 x_1^2} \otimes \sum_{a_1^2} \puretomixed{a_1^2} \otimes \tr_{E_A}(A_{x_1^2}^{E_A}(a_1^2) \rho^{(0)}_{E_A E}),
	\end{align*}
	and
	\begin{align*}
		\rho_{\Omega_1^2 X_1 A_1^2 E} &= \sum_{\omega_1^2 x_1^2} \rho(x^2_1) \rho(\omega_1|x_1)\rho(\omega_2|x_2) \puretomixed{\omega_1^2 x_1} \otimes \sum_{a_1^2} \puretomixed{a_1^2} \otimes \tr_{E_A}(A_{x_1^2}^{E_A}(a_1^2) \rho^{(0)}_{E_A E}) \\
		&= \sum_{\omega_1^2 x_1} \rho(x_1) \rho(\omega_1|x_1)\rho(\omega_2) \puretomixed{\omega_1^2 x_1} \otimes \sum_{a_1^2} \puretomixed{a_1^2} \otimes \tr_{E_A}(\sum_{x_2} \rho(x_2 | \omega_2) A_{x_1^2}^{E_A}(a_1^2) \rho^{(0)}_{E_A E}) \\
		&= \sum_{\omega_1^2 x_1} \rho(x_1) \rho(\omega_1|x_1)\rho(\omega_2) \puretomixed{\omega_1^2 x_1} \otimes \sum_{a_1^2} \puretomixed{a_1^2} \otimes \tr_{E_A}(A_{x_1 \omega_2}^{E_A}(a_1^2) \rho^{(0)}_{E_A E})
	\end{align*}
	where we define $A_{x_1 \omega_2}^{E_A}(a_1^2) := \sum_{x_2} \rho(x_2 | \omega_2) A_{x_1^2}^{E_A}(a_1^2)$. This state can also be written as
	\begin{align*}
		\rho_{X_1 \Omega_1 \Omega_2 A_1^2 E} = \sum_{ x_1} \rho(x_1) \puretomixed{ x_1} & \otimes \left( \sum_{\omega_1} \rho(\omega_1|x_1) \puretomixed{\omega_1} \right) \\
		& \otimes \left( \sum_{\omega_2} \rho(\omega_2) \puretomixed{\omega_2} \otimes\sum_{a_1^2} \puretomixed{a_1^2} \otimes \tr_{E_A}(A_{x_1 \omega_2}^{E_A}(a_1^2) \rho^{(0)}_{E_A E}) \right).
	\end{align*}
	This proves that $\Omega_1 \leftrightarrow X_1 \leftrightarrow \Omega_2 A_1^2 E$ form a Markov chain in the state $\rho_{X_1 \Omega_1 \Omega_2 A_1^2 E}$. 
\end{proof}
Note that for the variables in the lemma above, we also have $\Omega_1 \leftrightarrow \Omega_2 X_1 A_1 E \leftrightarrow  A_2 $ using properties of Markov chains. We will use this fact in the following.\\

Fix $I_j = i_j$. Define $C' := C \cup \{ i_j \}$. The state $\rho^{(0)}_{\Omega_{C'} \Omega_{C'^c} X_{C'} X_{C'^c} E_A E} := \rho_{\Omega_{C'} \Omega_{C'^c} X_{C'} X_{C'^c}} \otimes \psi_{E_A E}$ satisfies the conditions for Lemma \ref{lemm:MarkovChainOmegaXA}. Moreover, the register $E_A$ is measured using the measurement $A_{x_1^n}(\cdot)$, which only depends on $X_1^n$ to create $A_1^n$. Therefore, using the lemma above, we have that the state $\rho$ satisfies
\begin{align}
	& \Omega_{C'} \leftrightarrow \Omega_{C'^c} X_{C'} A_C E \leftrightarrow A_{C^c} \\
	\Rightarrow\ & \Omega_{C'} \leftrightarrow \Omega_{C'^c} X_{C'} A_C E \leftrightarrow A_{i_j}.
\end{align}
Since, this is true for every $i_j$, we have
\begin{align}
	\Omega_{C} \Omega_{I_j} \leftrightarrow I_j \Omega_{C'^c} X_{C} X_{I_j} A_C E \leftrightarrow A_{I_j}.
\end{align}
This allows us to bound
\begin{align*}
	H(A_{I_j} | E I_j \Omega_1^n A_C )_{\rho} & \geq H(A_{I_j} | E I_j \Omega_1^n  X_{C} X_{I_j} A_C)_{\rho} \\
	& = H(A_{I_j} | E I_j \Omega_{C'^c} X_{C} X_{I_j} A_C)_{\rho} \\
	& \geq  H(A_{I_j} | E I_j R_{-I_j} X_{I_j} )_{\rho}. \numberthis \label{eq:BoundToCont}
\end{align*}
In the second line above, we used the fact that if $A \leftrightarrow B \leftrightarrow C$, then $H(A|BC) = H(A|B)$. \\

\noindent \textbf{\normalsize Step 2: Simulating $\rho$ using a simple single-round strategy}\\

We will now show that it is possible to approximate the state $\rho_{I_j R_{-I_j} X_{I_j} Y_{I_j} A_{I_j} B_{I_j} E}$ using a quantum strategy for a single instance of the 3CHSH$_\perp$ game, which uses $(I_j , R_{-I_j})$ as shared classical random variables between Alice and Bob. This step is much more challenging and requires us to use deeper results from \cite{Bavarian21}. We use these results primarily to show that the random variable $R_{-I_j}$ is not too correlated with the answer $A_{I_j}$ and hence $H(A_{I_j} | E I_j R_{-I_j} X_{I_j} )_{\rho}$ is large.\\

Consider Proposition \ref{prop:unit_steering} applied to the subset $C \subseteq J$ defined above. Let $U_{r_{-i}, x}^{E_A}$ and $V_{r_{-i}, y}^{E_B}$ be the unitaries provided by this proposition. Then, we have that 
\begin{align}
	\Ex_{I}\Ex_{R_{-i}} \Ex_{XY} \norm{U_{r_{-i}, x}^{E_A} \otimes V_{r_{-i}, y}^{E_B} \otimes \Id^E \ket{\tilde{\Phi}^{(r_{-i}, \perp, \perp)}}_{E_A E_B E} - \ket{\tilde{\Phi}^{(r_{-i}, x, y)}}_{E_A E_B E}} = O(\delta^{1/16}/\alpha^3),
\end{align}
where $i$ is sampled uniformly at random from $C^c$, $r_{-i}$ sampled from $P_{R_{-i}}$ and $x,y$ are sampled from $P_{XY}$. This relation hints at a plausible simulation strategy: Eve samples and distributes $I$, $R_{-I}$ and the state $\ket{\Phi^{(r_{-i}, \perp, \perp)}}_{E_A E_B E}$ between Alice and Bob. Then, given questions $x$ and $y$ during the single-round game, Alice and Bob apply the unitaries $U_{r_{-i}, x}^{E_A}$ and $V_{r_{-i}, y}^{E_B}$ to their registers. This allows them to bring their shared state close to the state $\ket{\Phi^{(r_{-i}, x, y)}}_{E_A E_B E}$. Finally, Alice and Bob use appropriately defined measurements to sample their answers, simulating a state close to $\rho_{I_j R_{-I_j} X_{I_j} Y_{I_j} A_{I_j} B_{I_j} E}$ through a single-round strategy for 3CHSH$_\perp$. \\

The measurements used by Alice and Bob in this last step are defined as 
\begin{align}
	\hat{A}_{r_{-i}, x}(a_i) &:= {A}_{\omega_{-i}, x}(a_c)^{-1/2} \rndBrk{\sum_{a_1^n | a_i, a_C}  {A}_{\omega_{-i}, x}(a_1^n)} {A}_{\omega_{-i}, x}(a_C)^{-1/2} \\
	\hat{B}_{r_{-i}, y}(b_i) &:= {B}_{\omega_{-i}, y}(b_c)^{-1/2} \rndBrk{\sum_{b_1^n | b_i, b_C} {B}_{\omega_{-i}, y}(b_1^n)} {B}_{\omega_{-i}, y}(b_C)^{-1/2} 
\end{align}
where $r_{-i} = (\omega_{-i}, a_C, b_C)$. We state the simulation protocol for the state $\rho_{I_j R_{-I_j} X_{I_j} Y_{I_j} A_{I_j} B_{I_j} E}$ in Box \ref{box:rho_sim}. Let the state obtained through this simulation procedure be $\sigma$.\\
\begin{figure}
  \begin{mdframed}
  \textbf{Single-round protocol for simulating $\rho_{I_j R_{- I_j} X_{I_j} Y_{I_j} A_{I_j} B_{I_j} E}$:}
	\begin{enumerate}
		\item Eve chooses the random variable $I$ uniformly at random from $C^c$, the random variable $R_{-I}$ according to $P_{R_{-I}}$ depending on the value of $I$, and distributes both of them to Alice and Bob. 
		\item Eve distributes the registers $E_A$ and $E_B$ of the state $\ket{\tilde{\Phi}^{(r_{-i}, \perp, \perp)}}_{E_A E_B E}$ between Alice and Bob.
		\item Alice and Bob play the 3CHSH$_{\perp}$ game as follows:
		\begin{enumerate}
			\item Let $(i, r_{-i})$ be the classical variables provided to them by Eve in the first step, and let $x$ and $y$ be their questions. 
			\item Alice applies the unitary $U_{r_{-i}, x}^{E_A}$ to her register $E_A$. Similarly, Bob applies $V_{r_{-i}, y}^{E_B}$ to his register $E_B$. \label{st:sim_prot_AB_apply_UV} 
			\item Alice and Bob measure their registers with $\{\hat{A}_{r_{-i}, x}(a)\}_a$ and $\{\hat{B}_{r_{-i},y}(b) \}_b$ to generate their answers for the game. 
		\end{enumerate}
	\end{enumerate}
  \end{mdframed}
  {\captionof{SimBox}{}
  \label{box:rho_sim}}
\end{figure}

The proof that the state produced by the simulation in Box \ref{box:rho_sim} is close to $\rho_{I_j R_{-I_j} X_{I_j} Y_{I_j} A_{I_j} B_{I_j} E}$ parallels the arguments used to prove \cite[Lemma 6.2]{Bavarian21}. We cannot directly use the results in \cite[Section 6.1]{Bavarian21} because they only focus on the classical distributions of the answers, whereas we also need to take into account Eve's partial state. \\

The following lemma shows that if Alice and Bob share the state $\tilde{\Phi}_{E_A E_B E}^{(r_{-i} , x, y)}$ between them and measure it with the measurements $\hat{A}_{r_{-i} , x}$ and $\hat{B}_{r_{-i} , y}$, then the resulting answers and Eve's state are distributed as in the parallely repeated strategy conditioned on the classical variables $r_{-i} , x, y, a, b$. Its proof follows the proof of \cite[Claim 6.3]{Bavarian21}. 
\begin{lemma}
	Let $\rho_{E}^{(r_{-i} , x, y, a, b)}$ represent the state $\rho$ conditioned on the classical variables $R_{-i} = r_{-i}, X_i =x, Y_i =y, A_i=a$, and $B_i = b$, that is, 
	\begin{align}
		\rho_{E}^{(r_{-i}, x,y,a,b)} &= \Ex_{x_1^n, y_1^n, a_1^n, b_1^n | r_{-i}, x,y,a,b} \left[ \rho_{E}^{(x_1^n, y_1^n, a_1^n, b_1^n)} \right].
	\end{align}
	We have the equality
	\begin{align}
		\tr_{E_A E_B} \left( \hat{A}_{r_{-i} , x} (a) \otimes \hat{B}_{r_{-i} , y}(b) \tilde{\Phi}_{E_A E_B E}^{(r_{-i} , x, y)} \right) = P_{A_i B_i | R_{-i} X_i Y_i} (a,b | r_{-i}, x, y ) \rho_{E}^{(r_{-i} , x, y, a, b)}.
	\end{align}
\end{lemma}
\begin{proof}
	Using the definitions of $\tilde{\Phi}_{E_A E_B E}^{(r_{-i} , x, y)}$, $\hat{A}_{r_{-i} , x}(a)$ and $\hat{B}_{r_{-i} , y}(b)$, we have 
	\begin{align}
		\tr&_{E_A E_B} \left( \hat{A}_{r_{-i} , x} (a) \otimes \hat{B}_{r_{-i} , y}(b) \tilde{\Phi}_{E_A E_B E}^{(r_{-i} , x, y)} \right) \\
		&= \norm{\Phi^{(r_{-i} , x, y)}}^{-2} \tr_{E_A E_B} \big(\hat{A}_{r_{-i} , x} (a) \otimes \hat{B}_{r_{-i} , y}(b)  A_{\omega_{-i} , x} (a_C)^{1/2} \otimes B_{\omega_{-i}, y} (b_C)^{1/2} \Psi \nonumber \\
		&\hspace{8cm} A_{\omega_{-i}, x} (a_C)^{1/2} \otimes B_{\omega_{-i}, y} (b_C)^{1/2} \big)\\
		&= \norm{\Phi^{(r_{-i} , x, y)}}^{-2} \tr_{E_A E_B} \left( \sum_{a_1^n | a_i, a_C} A_{\omega_{-i}, x } (a_1^n) \otimes \sum_{b_1^n | b_i, b_C} B_{\omega_{-i},y } (b_1^n) \Psi \right) \\
		&= \norm{\Phi^{(r_{-i} , x, y)}}^{-2} \sum_{\substack{a_1^n | a_i, a_C \\ b_1^n | b_i, b_C}} \tr_{E_A E_B} \left( \Ex_{x_1^n |\omega_{-i}, x} A_{ x_1^n } (a_1^n) \otimes \Ex_{y_1^n |\omega_{-i}, y} B_{y_1^n} (b_1^n) \Psi \right) \\
		&= \norm{\Phi^{(r_{-i} , x, y)}}^{-2} \Ex_{x_1^n y_1^n | \omega_{-i}, x, y} \sum_{\substack{a_1^n | a_i, a_C \\ b_1^n | b_i, b_C}} \tr_{E_A E_B} \left(A_{ x_1^n } (a_1^n) \otimes B_{y_1^n} (b_1^n) \Psi \right) \label{eq:use_ques_ind_given_omega}\\
		&= \norm{\Phi^{(r_{-i} , x, y)}}^{-2} \Ex_{x_1^n y_1^n | \omega_{-i}, x, y} \sum_{\substack{a_1^n | a_i, a_C \\ b_1^n | b_i, b_C}} P_{A_1^n B_1^n | X_1^n Y_1^n}(a_1^n, b_1^n | x_1^n, y_1^n) \rho_{E}^{(x_1^n, y_1^n, a_1^n, b_1^n)} \\
		&= \norm{\Phi^{(r_{-i} , x, y)}}^{-2} \sum_{x_1^n y_1^n} \sum_{\substack{a_1^n | a_i a_C \\ b_1^n | b_i b_C}} P_{X_1^n Y_1^n | \Omega_{-i} X_i Y_i}(x_1^n, y_1^n|\omega_{-i}, x, y) P_{A_1^n B_1^n | X_1^n Y_1^n}(a_1^n, b_1^n | x_1^n, y_1^n) \rho_{E}^{(x_1^n, y_1^n, a_1^n, b_1^n)} \\
		&= \norm{\Phi^{(r_{-i} , x, y)}}^{-2} \sum_{x_1^n y_1^n} \sum_{\substack{a_1^n | a_i a_C \\ b_1^n | b_i b_C}} P_{X_1^n Y_1^n A_1^n B_1^n|\Omega_{-i} X_i Y_i}(x_1^n, y_1^n, a_1^n, b_1^n|\omega_{-i}, x, y) \rho_{E}^{(x_1^n, y_1^n, a_1^n, b_1^n)} \label{eq:use_ans_omega_markov_ch}\\
		&= \norm{\Phi^{(r_{-i} , x, y)}}^{-2} P_{A_C B_C A_i B_i |\Omega_{-i} X_i Y_i}(a_C, b_C, a_i, b_i |\omega_{-i}, x, y ) \nonumber\\
		& \qquad  \sum_{x_1^n y_1^n} \sum_{\substack{a_1^n | a_i a_C \\ b_1^n | b_i b_C}} P_{X_1^n Y_1^n A_1^n B_1^n |\Omega_{-i} X_i Y_i A_C B_C A_i B_i}(x_1^n, y_1^n, a_1^n, b_1^n|\omega_{-i}, x, y, a_C, b_C, a_i, b_i) \rho_{E}^{(x_1^n, y_1^n, a_1^n, b_1^n)} \\
		&= \norm{\Phi^{(r_{-i} , x, y)}}^{-2} P_{A_C B_C A_i B_i |\Omega_{-i} X_i Y_i}(a_C, b_C, a_i, b_i |\omega_{-i}, x, y ) \Ex_{x_1^n, y_1^n, a_1^n, b_1^n | r_{-i}, x,y,a_i ,b_i} \left[ \rho_{E}^{(x_1^n, y_1^n, a_1^n ,b_1^n)} \right] \\
		&= P_{A_C B_C |\Omega_{-i} X_i Y_i}(a_C, b_C| \omega_{-i}, x, y)^{-1} P_{A_C B_C A_i B_i |\Omega_{-i} X_i Y_i}(a_C, b_C, a_i, b_i |\omega_{-i}, x, y ) \rho_{E}^{(r_{-i} , x, y, a, b)} \\
		& = P_{A_i B_i | R_{-i} X Y} (a,b | r_{-i} x, y ) \rho_{E}^{(r_{-i} , x, y, a, b)}
	\end{align}
	where we have used the fact that $P_{X_1^n Y_1^n | \Omega_{-i}, X_i, Y_i} = P_{X_1^n| \Omega_{-i}, X_i} P_{Y_1^n | \Omega_{-i}, Y_i}$ in Eq. \ref{eq:use_ques_ind_given_omega} and the fact that given the questions the answers do not depend on $\Omega_{-i}$ in Eq. \ref{eq:use_ans_omega_markov_ch}. 
\end{proof}
\noindent Define the auxiliary state $\theta^{(0)}$ as 
\begin{align}
	\theta_{R_{-i} X_i Y_i E_A E_B E}^{(0)} :=& \sum_{r_{-i}, x, y }  P_{R_{-i}}(r_{-i}) P_{XY}(x,y) \puretomixed{r_{-i}, x, y} \otimes \tilde{\Phi}_{E_A E_B E}^{(r_{-i}, x, y)}.
\end{align}
One can view this as the state produced when $r_{-i}$ is sampled according to $P_{R_{-i}}$, $x,y$ are sampled according to the question distribution $P_{XY}$ and the state $\tilde{\Phi}_{E_A E_B E}^{(r_{-i}, x, y)}$ is distributed between Alice, Bob and Eve. We also define $\theta$ to be the state produced when Alice and Bob use the measurements $\hat{A}_{r_{-i}, x}$ and $\hat{B}_{r_{-i}, y}$ to measure $\theta^{(0)}$
\begin{align}
	\theta&_{R_{-i}X_i Y_i A_i B_i E} \nonumber\\
	:=& \sum_{r_{-i},x, y} P_{R_{-i}}(r_{-i}) P_{XY}(x,y) \puretomixed{r_{-i},x, y} \otimes \sum_{a, b} \puretomixed{a, b} \otimes \tr_{E_A E_B} \left( \hat{A}_{r_{-i}, x}(a) \otimes \hat{B}_{r_{-i}, y}(b)\tilde{\Phi}_{E_A E_B E}^{(r_{-i}, x, y)} \right) \label{eq:theta_defn1}\\
	=& \sum_{r_{-i},x, y} P_{R_{-i}}(r_{-i}) P_{XY}(x,y) \puretomixed{r_{-i},x, y} \otimes \sum_{a, b} P_{A_i B_i | R_{-i} X_i Y_i } (a,b |r_{-i}, x, y ) \puretomixed{a, b} \otimes \rho_{E}^{(r_{-i} , x, y, a, b)}. \label{eq:theta_defn2}
\end{align}
We will show that $\theta$ is close to both the simulated state and the real state. Using the triangle inequality, this will imply that the simulated state is also close to the real state. \\

Similar to $\theta^{(0)}$ above, we also define the state $\sigma^{(0)}$ after Step \ref{st:sim_prot_AB_apply_UV} in Box \ref{box:rho_sim} (conditioned on $I= i$) as
\begin{align}
	\sigma_{R_{-i} X_i Y_i E_A E_B E}^{(0)} :=& \sum_{r_{-i}, x, y }  P_{R_{-i}}(r_{-i}) P_{XY}(x,y) \puretomixed{r_{-i}, x, y} \otimes \rndBrk{{U}_{r_{-i}, x}^{E_A} \otimes {V}_{r_{-i}, y}^{E_B}} \tilde{\Phi}_{E_A E_B E}^{(r_{-i}, \perp, \perp)} \rndBrk{{U}_{r_{-i}, x}^{E_A\dagger} \otimes {V}_{r_{-i}, y}^{E_B\dagger}}.
\end{align}
The simulated state $\sigma_{R_{-i} X_i Y_i A_i B_i E}$ for $I=i$ is simply the state obtained when $\sigma^{(0)}$ is measured using the measurements $\hat{A}_{r_{-i} x}$ and $ \hat{B}_{r_{-i} y}$. \\

Using Proposition \ref{prop:unit_steering}, it is straightforward to show that for a random $i \in C^c$, $\theta$ and $\sigma$ are close to each other. The following lemma is essentially a generalisation of \cite[Claim 6.4]{Bavarian21}.
\begin{lemma}
	For the state $\theta$ defined in Eq. \ref{eq:theta_defn1} and the simulated state $\sigma$ (conditioned on $I=i$), we have that
	\begin{align*}
		\Ex_I \Vert \theta_{R_{-i} X_i Y_i A_i B_i E } - \sigma_{R_{-i} X_i Y_i A_i B_i E }\Vert_1 \leq O(\delta^{1/16}/\alpha^3)
	\end{align*}
	where $I$ is uniformly distributed at random in $C^c$. 
	\label{claim:theta-sigmadist}
\end{lemma}
\begin{proof}
	We have that 
	\begin{align*}
		\Ex_I \Vert \theta_{R_{-i} X_i Y_i A_i B_i E } - \sigma_{R_{-i} X_i Y_i A_i B_i E }\Vert_1 &\leq \Ex_I \Vert \theta_{R_{-i} X_i Y_i E_A E_B E }^{(0)} - \sigma_{R_{-i} X_i Y_i E_A E_B E }^{(0)} \Vert_1 \\
		& = \Ex_I \Ex_{R_{-i}} \Ex_{XY} \left[ \Vert \tilde{\Phi}_{E_A E_B E}^{(r_{-i}, x, y)} - \left( {U}_{r_{-i}, x}^{E_A} \otimes {V}_{r_{-i}, y}^{E_B} \right) \tilde{\Phi}_{E_A E_B E}^{(r_{-i}, \perp, \perp)} \left( {U}_{r_{-i}, x}^{E_A\dagger} \otimes {V}_{r_{-i}, y}^{E_B\dagger} \right)  \Vert_1 \right] \\
		& \leq O(\delta^{1/16}/\alpha^3)
	\end{align*}
	where the first line follows from the data processing inequality for norms, and the last line follows from Proposition \ref{prop:unit_steering} and the fact that $\norm{\psi - \phi}_1 \leq \sqrt{2}\norm{\ket{\psi} - \ket{\phi}}$ for all pure states $\psi$ and $\phi$. 
\end{proof}
Real state of the protocol $\rho$ conditioned on $I_j = i$ can be written as 
\begin{align*}
	\rho_{R_{-i} X_i Y_i A_i B_i E} &= \sum_{r_{-i} } P_{R_{-i}}(r_{-i}) \puretomixed{r_{-i}} \otimes  \sum_{x, y} P_{X_i Y_i | R_{-i}} (x, y |r_{-i}) \puretomixed{x, y} \\
	& \qquad \qquad \otimes \sum_{a, b} P_{A_i B_i | R_{-i} X_i Y_i} (a,b | r_{-i},x,y) \puretomixed{a, b} \otimes \rho_{E}^{(r_{-i} , x, y, a, b)}. \numberthis
\end{align*}
\begin{lemma}
	The real state of the protocol $\rho$ (conditioned on $I_j = i$) and the auxiliary state $\theta$ satisfy
	\begin{align*}
		\Ex_I \Vert \theta_{R_{-i} X_i Y_i A_i B_i E } - \rho_{R_{-i} X_i Y_i A_i B_i E } \Vert_1 \leq O(\delta_{}^{1/2}/\alpha^2)
	\end{align*}
	where $I$ is uniformly distributed at random in $C^c$. 
	\label{claim:theta-rhodist}
\end{lemma}
\begin{proof}
	\begin{align*}
		\Ex_I & \left[  \Vert \theta_{R_{-i} X_i Y_i A_i B_i E } - \rho_{R_{-i} X_i Y_i A_i B_i E } \Vert_1 \right]\\ 
		&= \Ex_I \left[ \lr{\Vert \sum_{r_{-i}, x,y} P_{R_{-i}}(r_{-i}) P_{XY} (x, y ) \puretomixed{r_{-i}, x, y} \otimes \rho_{A_i B_i E}^{(r_{-i}, x, y)} - \sum_{r_{-i}, x,y} P_{R_{-i}XY}(r_{-i}, x, y ) \puretomixed{r_{-i}, x, y} \otimes \rho_{A_i B_i E}^{(r_{-i}, x, y)} }\Vert_1 \right] \\
		&= \Ex_I \sqBrk{ \lr{\Vert P_{R_{-i}} P_{XY} - P_{R_{-i} X_i Y_i }}\Vert_1 }\\
		&\leq O(\delta^{1/2}/\alpha^2),
	\end{align*}
	where the last line follows from the equation after Eq. 88 in \cite{Bavarian21} (setting $W_C$ to be the trivial event).
\end{proof}
\begin{lemma}
	The state $\sigma$ produced by the single-round protocol for the 3CHSH$_\perp$ game in Box \ref{box:rho_sim} approximates the state $\rho$. Specifically, 
	\begin{align*}
		\lr{\Vert \rho_{I_j R_{-I_j} X_{I_j} Y_{I_j} A_{I_j} B_{I_j} E } -  \sigma_{I R_{-I} X_{I} Y_{I} A_{I} B_{I} E }}\Vert_1 &= O(\delta^{1/16}/\alpha^3).
	\end{align*}
	\label{lemm:SimulDist}
\end{lemma}
\begin{proof}
	Using the fact that $I_j$ in $\rho$ is distributed uniformly at random in $C^c$, and Lemma \ref{claim:theta-sigmadist} and \ref{claim:theta-rhodist}, we have  
	\begin{align*}
		& \lr{\Vert \rho_{I_j R_{-I_j} X_{I_j} Y_{I_j} A_{I_j} B_{I_j} E } -  \sigma_{I R_{-I} X_{I} Y_{I} A_{I} B_{I} E }}\Vert_1 =\Ex_{I} \lr{\Vert \rho_{R_{-i} X_{i} Y_{i} A_{i} B_{i} E } -  \sigma_{R_{-i} X_{i} Y_{i} A_{i} B_{i} E }}\Vert_1 \\
		& \leq \Ex_I \Vert \theta_{R_{-i} X_i Y_i A_i B_i E } - \rho_{R_{-i} X_i Y_i A_i B_i E } \Vert_1 + \Ex_I \Vert \theta_{R_{-i} X_i Y_i A_i B_i E } - \sigma_{R_{-i} X_i Y_i A_i B_i E }\Vert_1 \\
		& \leq O(\delta^{1/16}/\alpha^3)
	\end{align*}
	where $\Ex_I$ represents expectation over $I$, which is sampled uniformly at random from $C^c$. 
\end{proof}

\noindent \textbf{\normalsize Step 3: Bounding $H(A_{I_j} | E I_j R_{-I_j} X_{I_j} )_{\rho}$ using the single-round entropy bound}\\

Using Lemma \ref{lemm:SimulDist} above, we show that the single-round strategy in Box \ref{box:rho_sim} has a large winning probability if $\rho$ has a large average winning probability. This will be helpful for bounding the entropy of the answers of the game given Eve's register while using the single-round bound in Lemma \ref{lemm:3CHSH_anch_entropy}. \\

Let $W_i$ denote the indicator random variable for the event that game $G_i$ is won by the players. Since, $W_i$ is a deterministic function of $X_i, Y_i, A_i, B_i$ we have
\begin{align}
\vert \Pr_{\rho}( W_{I_j}) -  \Pr_\sigma(W_{I})\vert &\leq \lr{\Vert \rho_{I_j R_{-I_j} X_{I_j} Y_{I_j} A_{I_j} B_{I_j} W_{I_j} E } - \sigma_{I R_{-I} X_{I} Y_{I} A_{I} B_{I} W_I E } }\Vert_1 \leq O \rndBrk{ \delta_{}^{1/16}/\alpha^3}.
\label{eq:winning_prob_reln}
\end{align}
For the state, $\rho$, $I_j$ is an index chosen uniformly at random in $C^c$. Using Eq. \ref{eq:HighAvgWinningProb}, we have
\begin{align*}
	\Pr_{\rho}(W_{I_j}) & \geq (1-\epsilon) \frac{\omega_{\text{th}} n - |C|}{n-|C|} \\
	&\geq (1-\epsilon) \rndBrk{\omega_{\text{th}} - \delta} \numberthis
\end{align*}
Combining this with Eq. \ref{eq:winning_prob_reln}, we get that 
\begin{align}
	\Pr_\sigma(W_{I}) \geq \omega_{\text{th}} - O \rndBrk{ \epsilon + \delta_{}^{1/16}/\alpha^3}
\end{align}
We now bound $H(A_{I_j} | E I_j R_{-I_j} X_{I_j} )_{\rho}$ using the entropy bound for the single-round 3CHSH$_\perp$ game. To use Lemma \ref{lemm:3CHSH_anch_entropy}, we first use the Alicki-Fannes-Winter (AFW) continuity bound for the conditional entropy \cite{Alicki04,Winter16} to lower bound the entropy on $\rho$ with the corresponding entropy on the simulated state $\sigma$:
\begin{align*}
	H(A_{I_j} | E I_j R_{-I_j} X_{I_j} )_{\rho} &\geq H(A_{I} | E I R_{-I} X_{I} )_{\sigma} - g(O(\delta^{1/16}/\alpha^3)) \\
	&\geq F\rndBrk{g_{\alpha, \nu}(\Pr_{\sigma} [W])} - g(O(\delta^{1/16}/\alpha^3)) \\
	&\geq F\rndBrk{g_{\alpha, \nu}( \omega_{\text{th}} - O(\epsilon + \delta^{1/16}/\alpha^3) )} - O\rndBrk{ \frac{\delta^{1/16}}{\alpha^3 }\log \frac{1}{\delta}} \\
	&\geq  F\rndBrk{g_{\alpha, \nu}( \omega_{\text{th}})} -  O((F \circ g_{\alpha, \nu})'(\omega_{\text{th}}) (\epsilon + \delta^{1/16}/\alpha^3)- O\rndBrk{ \frac{\delta^{1/16}}{\alpha^3 }\log \frac{1}{\delta}} \\
	&\geq  F\rndBrk{g_{\alpha, \nu}( \omega_{\text{th}})} - O\rndBrk{ \frac{\epsilon}{\nu}+ \frac{\delta^{1/16}}{\alpha^3 \nu }\log \frac{1}{\delta}}
\end{align*}
where $g(x) = 2x \log(|A|)+ (x+1) \log(x+1) - x \log(x) = O\rndBrk{x \log \frac{|A|}{x}}$. We use the single-round entropy bound (Lemma \ref{lemm:3CHSH_anch_entropy}) for the 3CHSH$_\perp$ game in the second line. This is valid because in the simulation of $\sigma$, Eve classically distributes $I$ and $R_{-I}$ to the players, so we can assume that she also holds copies of these registers. $(F \circ g_{\alpha, \nu})'$ represents the derivative of the function $F \circ g_{\alpha, \nu}$ with respect to $\omega$. We have used the fact that $F \circ g_{\alpha, \nu}$ is convex and increasing, and the bound $(F \circ g_{\alpha, \nu})'(\omega_{\text{th}}) \leq O(\frac{1}{\nu})$ above.\\

Finally, if we plug in the bound above in Eq. \ref{eq:chain_rule_for_privacy_amplification}, we get 
\begin{align}
	H(A_J | E \Omega_1^n J)_{\rho} &\geq \sum_{k=1}^t \Ex_{i_1^{k-1}} \left[ H(A_{I_k} | E I_k R_{-I_k} X_{I_k} )_{\rho} \right] \nonumber \\
	&\geq t \cdot \rndBrk{F\rndBrk{g_{\alpha, \nu} ( \omega_{\text{th}})} - O\rndBrk{ \frac{\epsilon}{\nu}+ \frac{\delta^{1/16}}{\alpha^3 \nu }\log \frac{1}{\delta}}}
\end{align}

\subsection{Bounding information reconciliation cost}
\begin{sloppypar}
To take the information reconciliation cost into account, we consider the quantity $H(A_J | B_J J)_{\rho_{\text{honest}}}$ for the honest protocol (with no Eve) under noisy conditions. We model the noise between Alice and Bob as a depolarising channel with noise parameter $2Q$. If Alice sends one half of a perfect Bell state to Bob for each round in the honest protocol, then their shared state is $\eta_{E_A E_B}^{\otimes n}$ where 
\begin{align}
	\eta_{E_A E_B} = (1-2Q) \ket{\Phi^+}\bra{\Phi^+} + 2Q \tau_{E_A E_B} 
\end{align}
where $\ket{\Phi^+}$ is a Bell state and $\tau_{E_A E_B}$ is the completely mixed state. The noise parameter, $Q$ can be measured using the qubit error rate through the equation \cite[Section 4.2.4]{Friedman20}
\begin{align}
	\Pr[A \neq B | X, Y = (0,2)] = Q.
\end{align}
We will assume that $Q \leq 0.1$. With probability at least $(1-\alpha)^2 (1-\nu) (1-Q)$, Alice and Bob's answers are equal in every round. Since the shared state is i.i.d, we can simply bound the entropy as
\begin{align*}
	H(A_J | B_J J)_{\rho_{\text{honest}}} &= t \cdot H(A | B)_{\eta} \\
	& \leq t \cdot h(1 - (1-\alpha)^2 (1-\nu) (1-Q)) \\
	& \leq t \cdot h(2\alpha + \nu + Q)
\end{align*}
where the last line is true for $\alpha, \nu, Q \in (0, 0.1)$.


\subsection{Bounding key length}
We can now estimate the length of the key produced during the parallel DIQKD protocol as:
\begin{align*}
	H(A_J &| E J \Omega_1^n )_{\rho} - H(A_J | B_J J)_{\rho_{\text{honest}}} \\
	&\geq t \rndBrk{ F\rndBrk{g_{\alpha, \nu} ( \omega_{\text{th}})} - O\rndBrk{ \frac{\epsilon}{\nu}+ \frac{\delta^{1/16}}{\alpha^3 \nu }\log \frac{1}{\delta}} - h(2\alpha + \nu + Q) }
\end{align*}
where $t = \frac{\delta}{\log |\mathcal{A}| |\mathcal{B}| + \delta} n = \Omega(n)$. Note that we can choose $\omega_{\text{th}}$ close to the maximum probability of winning so that $F\circ g_{\alpha, \nu} ( \omega_{\text{th}})$ is at least some constant, say $0.5$. $\alpha, \nu$ and $Q$ can be chosen close to zero so that the information reconciliation term is small. Finally, with these parameters fixed, we can choose $\delta$ small enough so that the key length is $\Omega(n)$ and the key rate is positive. This completes the proxy von Neumann entropy based security proof. 

\end{sloppypar}

\section{Security proof}
\label{sec:one_shot_pf}

The von Neumann entropy based security argument illuminates the fundamental mechanism behind entropy accumulation in parallel DIQKD. It demonstrates that Alice's answers $A_J$ are random with respect to the adversary Eve because, for every $k \in [t]$, one can approximate the questions and answers in the partial state $\rho_{I_j R_{-I_j} X_{I_j} Y_{I_j} A_{I_j} B_{I_j} E}$ as the output of a single-round strategy for the 3CHSH$_\perp$ game, which has a high winning probability. Now, we need to port the lower bound on the von Neumann entropy of Alice's answers to a smooth min-entropy lower bound. To chart the course forward, we can compare this situation with sequential DIQKD, where Alice's answers in each round directly result from a single-round strategy. This direct relationship ensures their randomness according to the single-round entropy bound. Consequently, when the average winning probability is high, entropy accumulates across all rounds. The entropy accumulation theorem (EAT) \cite{Dupuis20} (also see \cite{Metger22}) serves as the primary information-theoretic tool for demonstrating this accumulation in the security proof for sequential DIQKD \cite{Friedman18,Arnon_Friedman19}. \\

EAT provides a method to bound the smooth min-entropy of a state generated by an $n$-step sequential process. This process, illustrated in Fig. \ref{fig:EAT_setup}, produces secret information $A_k$ and side information $B_k$ at each step. In the context of sequential DIQKD, the channels $\cM_k$ represent the operations performed by Alice and Bob to generate questions (constituting the register $B_k$) and answers (forming the register $A_k$) for the 3CHSH game. EAT lower bounds the smooth min-entropy of the secret information (answers) with respect to the side information (questions) and Eve's initial register, provided the side information meets certain independence criteria. For the sequential DIQKD protocol, this is essentially the smooth min-entropy bound required for proving security. \\

\begin{figure}
    \centering
    \includegraphics[scale=0.1]{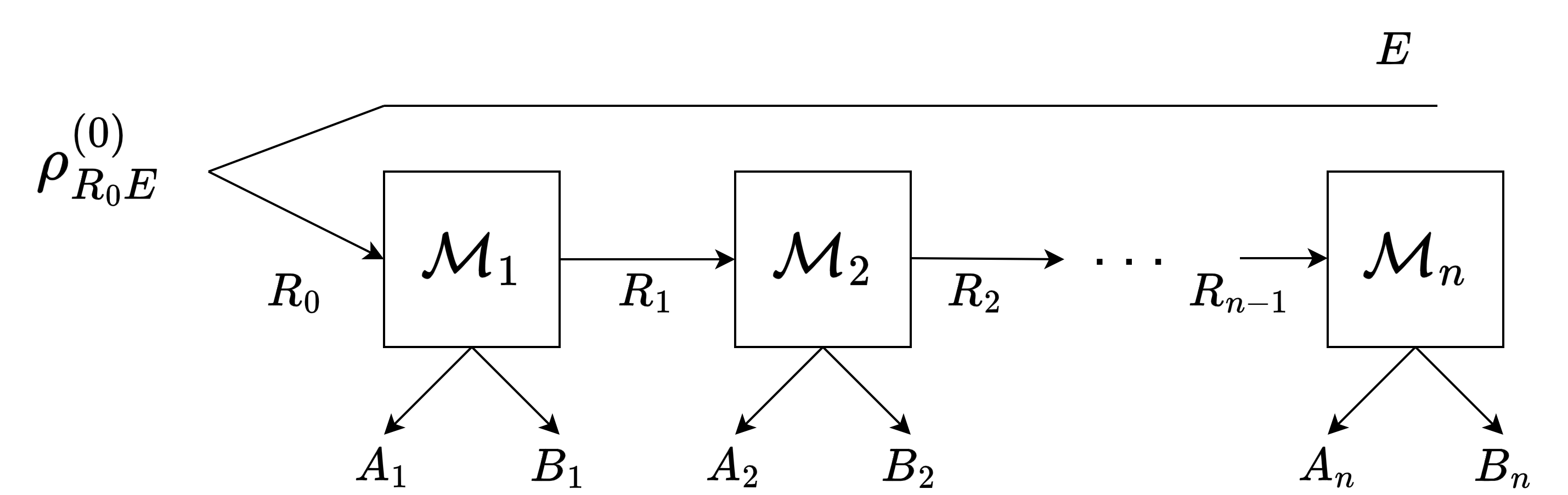}
    \caption{The setting for entropy accumulation. The channels $\cM_k$ are sequentially applied to the registers $R_{k-1}$ and produce the secret information $A_k$ and the side information $B_k$.}
    \label{fig:EAT_setup}
  \end{figure}

To prove security for the parallel DIQKD protocol, we require a tool analogous to EAT. Drawing insights from the von Neumann entropy based argument, we can identify key features necessary for such a tool. The most crucial distinction between the setting in Fig. \ref{fig:EAT_setup} and the state $\rho$ produced in parallel DIQKD is the absence of a sequential or structured process generating the answers in $\rho$. Instead, Alice and Bob produce answers \emph{in parallel} using a single measurement channel. The primary source of structure, and the reason for expecting entropy accumulation, lies in our ability to approximate Alice and Bob's questions and answers on round $I_k$ along with Eve's information using the output of a single-round strategy. Consequently, the entropic tool we employ should not rely on a sequential structure for the state. Rather, it should demonstrate entropy accumulation when the partial state of the given state can be approximated as the output of a suitable channel. To address this need, we developed an \emph{unstructured} approximate entropy accumulation theorem in our companion paper \cite{Marwah24_uni_ch}.\\

This theorem shows that for any state $\rho_{A_1^n B_1^n E}$ whose partial states $\rho_{A_1^k B_1^k E}$ can be $\epsilon$-approximated as the output of channels $\cM_k$, which sample the side information $B_k$ independent of the previous registers $A_1^{k-1} B_1^{k-1} E$, we can recover a statement similar to EAT. Crucially, Theorem \ref{th:approx_EAT_wtest} does not require $\rho_{A_1^n B_1^n E}$ to be produced by a structured process like Fig. \ref{fig:EAT_setup}. In essence, the state may even be produced by a completely parallel process, as is the case in our scenario. \\

Incorporating approximations into EAT itself presents significant challenges. We previously developed another approximate entropy accumulation theorem in \cite[Theorem 5.1]{Marwah24_approx_ch}. Unlike the unstructured approximate EAT, this earlier theorem applies only to sequential processes of the form shown in Fig. \ref{fig:EAT_setup}. Moreover, it considers much stricter channel approximations to the channels $\cM_k$ forming this process. These limitations prevent us from using \cite[Theorem 5.1]{Marwah24_approx_ch} in our security proof. The unstructured approximate EAT's relaxed assumptions come with an inherent limitation: the smoothing parameter must depend on the approximation parameter, unlike in \cite[Theorem 5.1]{Marwah24_approx_ch} where it could be chosen arbitrarily. Nevertheless, in our security proof for parallel DIQKD, this limitation is not significant since we can choose the approximation parameter using the protocol parameter $\delta$. \\

In the following section, we state the unstructured approximate EAT. Subsequently, we apply this theorem to prove security for parallel DIQKD in the one-shot regime.

\subsection{Unstructured approximate entropy accumulation theorem}

In this section, we present the unstructured approximate entropy accumulation theorem (EAT) \cite{Marwah24_uni_ch}, so that we can use it to prove security for parallel DIQKD. Before stating the theorem, we will introduce the fundamental concepts and notations it employs. \\

We begin by defining the testing channels $\mathcal{T}_k$. These channels measure the outputs $A_k$ and $B_k$ of the state $\rho$ and output a result $X_k$ based on these measurements. Concretely, for every $k \in [n]$ the channel $\mathcal{T}_k : A_k B_k \rightarrow A_k B_k X_k$ is of the form
\begin{align}
    \mathcal{T}_k (\omega_{A_k B_k}) = \sum_{a, b} \Pi_{A_k}^{(a)} \otimes \Pi_{B_k}^{(b)} \omega_{A_k B_k} \Pi_{A_k}^{(a)} \otimes \Pi_{B_k}^{(b)} \otimes \ket{x(a,b)}\bra{x(a,b)}_{X_k}
    \label{eq:test_maps}
\end{align}
where $\{\Pi_{A_k}^{(a)}\}_a$ and $\{\Pi_{B_k}^{(b)}\}_b$ are orthogonal projectors and $x(\cdot)$ is some deterministic function which uses the measurements $a$ and $b$ to create the output register $X_k$.\\


Next we define the min-tradeoff functions. We consider $n$ channels $\cN_k$ for $k \in [n]$ and assume that the registers $X_k$, which consist of the result of the testing channels $\mathcal{T}_k$ are isomorphic, that is, $X_k \equiv X$ for all $k$. Let $\mathbb{P}$ be the set of probability distributions over the alphabet of $X$ registers. Let $R$ be any register isomorphic to $R_{k}$. For a probability $q \in \mathbb{P}$ and a channel $\cN_{k} : R_{k} \rightarrow A_k B_k$, we define the set 
\begin{align}
    \Sigma (q | \cN_{k}) := \curlyBrk{\nu_{A_k B_k X_k R} = \mathcal{T}_k \circ \cN_{k}(\omega_{R_{k}R}): \text{ for a state } \omega_{R_{k-1}R} \text{ such that }\nu_{X_k}= q}.
    \label{eq:Sigma_set_defn}
\end{align}

\begin{definition}
    A function $f: \mathbb{P} \rightarrow \mathbb{R}$ is called a min-tradeoff function for the channels $\{\cN_k \}_{k=1}^n$ if for every $k$ and $q \in \mathbb{P}$, it satisfies
    \begin{align}
        f(q) \leq \inf_{\nu \in \Sigma (q| \cN_{k})} H(A_k | B_k R)_{\nu}.
    \end{align}
\end{definition}
We now state the unstructured approximate EAT with testing. 
\begin{theorem}
  \label{th:approx_EAT_wtest}
    Let $\epsilon \in (0,1)$ and for every $k \in [n]$, the registers $A_k$ and $B_k$ be such that $|A_k| = |A|$ and $|B_k| = |B|$. Suppose, the state $\rho_{A_1^n B_1^n X_1^n E}$ is such that 
    \begin{enumerate}
      \item The registers $X_1^n$ can be recreated by applying the testing maps to the registers $A_1^n$ and $B_1^n$, that is, 
      \begin{align}
        \rho_{A_1^n B_1^n X_1^n E} = \mathcal{T}_n \circ \cdots \circ \mathcal{T}_1 (\rho_{A_1^n B_1^n E})
      \end{align}
      \item For every $k \in [n]$, there exists a channel $\cM_k: R_{k} \rightarrow A_k B_k$ and a state $\theta_{B_k}^{(k)}$ such that 
      \begin{align}
        &\tr_{X_k}\circ \mathcal{T}_k \circ \cM_k = \cM_k \\
        &\tr_{A_k}\circ \cM_k (X_{R_k}) = \tr(X) \theta_{B_k}^{(k)} \quad \text{ for all operators } X_{R_k}
      \end{align}
      and a state $\tilde{\rho}^{(k, 0)}_{A_1^{k-1} B_1^{k-1} E R_{k}}$ for which 
      \begin{align}
        \frac{1}{2}\norm{\rho_{A_1^k B_1^k E} - \cM_k(\tilde{\rho}^{(k, 0)}_{A_1^{k-1} B_1^{k-1} E R_{k}})}_1\leq \epsilon.
      \end{align} 
    \end{enumerate}
    Then, for an event $\Omega$ defined using $X_1^n$, an affine min-tradeoff function $f$ for $\{\cM_k\}_{k=1}^n$ such that for every $x_1^n \in \Omega$, $f(\text{freq}(x_1^n)) \geq h$, we have
    \begin{align}
      H_{\min}^{\mu' + \epsilon'}(A_1^n|B_1^n E)_{\rho_{|\Omega}} &\geq n(h  - V(3\sqrt{\mu} + 4\epsilon) -g_2 (2\epsilon)) \nonumber \\
      & \quad - \frac{V}{\sqrt{\mu}}\rndBrk{2\log\frac{1}{P_{\rho}(\Omega) - \mu} + \frac{2}{\mu^2} + 2 \log \frac{1}{1- \mu^2} + g_1(\epsilon', \mu')}
    \end{align}
    where 
    \begin{align}
      &\mu := \rndBrk{\frac{8 \sqrt{\epsilon} + 2\epsilon}{1- \epsilon^2/(|A| |B|)^2} \log\frac{|A| |B|}{\epsilon}}^{1/3}\\
      &\mu' := 2\sqrt{\frac{\mu}{P_{\rho}(\Omega)}} \\
      &V := \log \rndBrk{1 + 2|A|} + 2\lceil \norm{\nabla f}_\infty \rceil\\
      &g_1(x, y):= - \log(1- \sqrt{1-x^2}) - \log (1-y^2) \\
      &g_2(x) :=  x\log \frac{1}{x} + (1+x)\log (1+x)
    \end{align}
    and $\epsilon' \in (0,1)$ such that $\mu' + \epsilon' <1$.
\end{theorem}

\subsection{Using unstructured approximate EAT for the smooth min-entropy bound}

In order to prove security for Protocol \ref{prot:par_diqkd_prot}, we need to bound the smooth min-entropy of Alice's raw key with respect to Eve's information, that is, we need to lower bound 
\begin{align}
    H^{\epsilon}_{\min}(A_J | E \Omega_1^n J T_1^t X_S A_S)_{\rho_{|\lnot F}}.
    \label{eq:one_shot_target}
\end{align}
Recall that $J = \{I_1, I_2, \cdots, I_t\}$. For simplicity, for a sequence of variables $V_1^n$ (like $X_1^n$, $A_1^n$ etc.) and $j \in [t]$, we define the variable 
\begin{align}
    \hat{V}_j := V_{I_j}.
\end{align}
For example the entropy above in Eq. \ref{eq:one_shot_target} can be written as $H^{\epsilon}_{\min}(\hat{A}_1^t | E \Omega_1^n I_1^t T_1^t X_S A_S)_{\rho_{|\lnot F}}$. We will begin by using the unstructured approximate EAT to prove that  
\begin{align}
    H^{\epsilon}_{\min}(\hat{A}_1^t \hat{B}_1^t | \hat{X}_1^t \hat{Y}_1^t T_1^t I_1^t E \Omega_{J^c})_{\rho_{|\lnot F}}\geq \Omega(t).
\end{align} 
Through a minor modification of the simulation in Box \ref{box:rho_sim}, we will show that the states $\rho_{\hat{A}_1^j \hat{B}_1^j \hat{X}_1^j \hat{Y}_1^j T_1^j I_1^t E \Omega_{J^c}}$ can be approximated by states, which are produced by playing a single-round 3CHSH$_\perp$ game. This modified simulation is presented in Box \ref{box:one_shot_sigma}. Let $\sigma^{(j)}$ denote the state produced by this simulation. 

\begin{figure}
    \begin{mdframed}
        \textbf{Single-round protocol for producing $\sigma^{(j)}_{\hat{A}_1^j \hat{B}_1^j \hat{X}_1^j \hat{Y}_1^j I_1^t T_1^j E \Omega_{J^c}}$:}
        \begin{enumerate}
            \item Eve randomly samples $J := \{I_1, I_2, \cdots, I_{t} \} \subseteq [n]$ of size $t$. Define $C:=I_1^{j-1}$. She also samples $T_1^{j-1}$ i.i.d from $\{0,1\}$ with $\Pr(T_i=1) = \gamma$.
            \item Eve randomly samples the random variable $R_{-I_j}= (A_C, B_C, X_C, Y_C, \Omega_{(C\cup\{i\})^c}) = (\hat{A}_1^{j-1}, \hat{B}_1^{j-1}, \hat{X}_1^{j-1}, \hat{Y}_1^{j-1}, \Omega_{J^c}, \hat{\Omega}_{j+1}^t)$ according to $P_{R_{-I_j}}$ depending on the value of $I_j$.
            \item Eve distributes $C,I_j, R_{-I_j}$ to both Alice and Bob. Call their copies $C^{(A)},I_j^{(A)}, R_{-I_j}^{(A)}$ and $C^{(B)},I_j^{(B)}, R_{-I_j}^{(B)}$.  
            \item Let $I_j = i $ and $R_{-I_j} = r_{-i}$. Eve distributes the registers $E_A$ and $E_B$ of the state $\ket{\tilde{\Phi}^{(r_{-i}, \perp, \perp)}}_{E_A E_B E}$ between Alice and Bob.
            \item Alice and Bob play the 3CHSH$_{\perp}$ (EAT map):
            \begin{enumerate}
                \item The questions $x$ and $y$ are sampled according to $P_{XY}$ and sent to Alice and Bob. 
                \item Let $U_{r_{-i}, x}^{E_A}$ and $V_{r_{-i}, y}^{E_B}$ be the unitaries defined by Proposition \ref{prop:unit_steering}. Alice applies the unitary $U_{r_{-i}, x}$ to her register $E_A$. Similarly, Bob applies $V_{r_{-i}, y}$ to his register $E_B$. 
                \item Alice and Bob measure their registers with $\{\hat{A}_{r_{-i}, x}(a)\}_a$ and $\{\hat{B}_{r_{-i},y}(b) \}_b$ to generate their answers for the game. 
                \item Alice randomly samples $T_{j} \in \{0,1\}$ with $\Pr(T_j=1) = \gamma$.
            \end{enumerate}
            \label{st:AB_play_CHSH}
            \item Eve traces over $\hat{\Omega}_{j+1}^t$ in her $R_{-I_j}$.
            \label{st:last_step}
        \end{enumerate}
    \end{mdframed}
    {\captionof{SimBox}{}
    \label{box:one_shot_sigma}}
\end{figure}

\begin{claim}
    For every $j \in [t]$, the state $\sigma^{(j)}$ produced through Box \ref{box:one_shot_sigma} satisfies
    \begin{align}
        \norm{\rho_{\hat{A}_1^j \hat{B}_1^j \hat{X}_1^j \hat{Y}_1^j T_1^j I_1^t E \Omega_{J^c}} - \sigma^{(j)}_{\hat{A}_1^j \hat{B}_1^j \hat{X}_1^j T_1^j \hat{Y}_1^j I_1^t E \Omega_{J^c}}}_1 \leq O\rndBrk{\frac{\delta^{1/16}}{\alpha^3}}.
    \end{align}
    \label{claim:rho_sigma_approx}
\end{claim}
\begin{proof}
    In the first step of Box \ref{box:one_shot_sigma}, condition on the choice of the indices $I_1^{j-1} = i_1^{j-1}$. This fixes the subset $C \subseteq J$. For this choice of $C$, consider the state $\sigma^{(j)}_{I_j R_{-I_j} X_{I_j} Y_{I_j} A_{I_j} B_{I_j} E | C= i_1^{j-1}}$ produced by the process in Box \ref{box:one_shot_sigma} before Step \ref{st:last_step} (ignoring $I_{j+1}^t$ and $T_1^j$ for now). Observe that this state is identical to the state produced using the simulation in Box \ref{box:rho_sim} for the fixed subset $C$. Using Lemma \ref{lemm:SimulDist}, we have that 
    \begin{align}
        \norm{\rho_{I_j R_{-I_j} X_{I_j} Y_{I_j} A_{I_j} B_{I_j} E | C= i_1^{j-1}} - \sigma^{(j)}_{I_j R_{-I_j} X_{I_j} Y_{I_j} A_{I_j} B_{I_j} E | C= i_1^{j-1}}}_1 \leq O\rndBrk{\frac{\delta^{1/16}}{\alpha^3}}.
    \end{align}
    Now, in the real protocol the subset $C = I_1^{j-1}$ is distributed as a random $j-1$ size subset of $[n]$. This is exactly how $I_1^{j-1}$ are also distributed in the simulation protocol in Box \ref{box:one_shot_sigma}. Therefore, we have 
    \begin{align}
        \norm{\rho_{I_1^j R_{-I_j} X_{I_j} Y_{I_j} A_{I_j} B_{I_j} E } - \sigma^{(j)}_{I_1^j R_{-I_j} X_{I_j} Y_{I_j} A_{I_j} B_{I_j} E }}_1 \leq O\rndBrk{\frac{\delta^{1/16}}{\alpha^3}}.
    \end{align} 
    We also need to incorporate ${I}_{j+1}^t$ in the inequality above. Let $\Phi : I_1^j \rightarrow I_1^t$ be the map which simply reads $I_1^j$ and randomly samples $I_j^{t}$ outside this set in $[n]$. This channel produces the correct distribution on $I_1^t$ when applied to both $\rho_{I_1^j}$ and $\sigma^{(j)}_{I_1^j}$. Further, observe that both $\rho$ and $\sigma^{(j)}$ satisfy the Markov chains $I_{j+1}^t \leftrightarrow I_1^j \leftrightarrow R_{-I_j} X_{I_j} Y_{I_j} A_{I_j} B_{I_j} E$. Therefore, using the data processing inequality for the trace norm we have 
    \begin{align}
        &\norm{\rho_{I_1^t R_{-I_j} X_{I_j} Y_{I_j} A_{I_j} B_{I_j} E } - \sigma^{(j)}_{I_1^t R_{-I_j} X_{I_j} Y_{I_j} A_{I_j} B_{I_j} E }}_1 \nonumber\\
        &\qquad\qquad= \norm{\Phi\rndBrk{\rho_{I_1^j R_{-I_j} X_{I_j} Y_{I_j} A_{I_j} B_{I_j} E }} - \Phi\rndBrk{\sigma^{(j)}_{I_1^j R_{-I_j} X_{I_j} Y_{I_j} A_{I_j} B_{I_j} E }} }_1 \nonumber \\
        &\qquad\qquad\leq O\rndBrk{\frac{\delta^{1/16}}{\alpha^3}}.
    \end{align}
    We can now also consider the last step in Box \ref{box:one_shot_sigma} by tracing over $\hat{\Omega}_{j+1}^t$. This gives us 
    \begin{align}
        \norm{\rho_{I_1^t \hat{X}_1^j \hat{Y}_1^j \hat{A}_1^j \hat{B}_1^j \Omega_{J^c} E } - \sigma^{(j)}_{I_1^t \hat{X}_1^j \hat{Y}_1^j \hat{A}_1^j \hat{B}_1^j \Omega_{J^c} E}}_1 \leq O\rndBrk{\frac{\delta^{1/16}}{\alpha^3}}.
    \end{align}
    Lastly, we also need to account for $T_1^j$. These are sampled independently in both $\rho$ and $\sigma^{(j)}$ with the same distribution. Hence, we can account for these simply and the claim follows from the bound above.  
\end{proof}

We will now apply Theorem \ref{th:approx_EAT_wtest} to the state $\rho$. The approximation chain given by the simulation procedure in Box \ref{box:one_shot_sigma} is useful for this purpose. For $j \in [t]$, let's define 
\begin{align}
    \tilde{E}^{(j)}_A :=  E_A I_j^{(A)} C^{(A)} R_{-I_j}^{(A)} \text{ and } \tilde{E}^{(j)}_B := E_B I_j^{(B)} C^{(B)} R_{-I_j}^{(B)} 
\end{align}
in Box \ref{box:one_shot_sigma} to be the entirety of Alice and Bob's registers before they start playing the 3CHSH$_\perp$ game in Step \ref{st:AB_play_CHSH} of Box \ref{box:one_shot_sigma}. Let $\mathcal{M}_j: \tilde{E}^{(j)}_A \tilde{E}^{(j)}_B  \rightarrow \hat{X}_j \hat{Y}_j T_j \hat{A}_j \hat{B}_j$ be the channel applied by Alice and Bob in Step \ref{st:AB_play_CHSH}. This channel samples the questions $\hat{X}_j$ and $\hat{Y}_j$ for the 3CHSH$_\perp$ game, applies Alice and Bob's measurements on $\tilde{E}^{(j)}_A$ and $\tilde{E}^{(j)}_B$ to produce their answers and also randomly samples $T_j$. Note that the Steps \ref{st:AB_play_CHSH} and \ref{st:last_step} commute, so we can change their order without affecting Claim \ref{claim:rho_sigma_approx}. For the following assume that Eve performs Step \ref{st:last_step} before Alice and Bob play the game (Step \ref{st:AB_play_CHSH}). Finally, let the state ${\sigma}^{(j,0)}_{\hat{A}_{1}^{j-1} \hat{B}_{1}^{j-1} \hat{X}_{1}^{j-1} \hat{Y}_{1}^{j-1} T_1^{j-1} I_1^t \tilde{E}^{(j)}_A \tilde{E}^{(j)}_B E \Omega_{J^c}}$ be the state between Alice, Bob and Eve before the 3CHSH$_\perp$ game is played. With these definitions, we have that for every $j \in [t]$
\begin{align}
    \sigma^{(j)}_{\hat{A}_1^j \hat{B}_1^j \hat{X}_1^j \hat{Y}_1^j T_1^j I_1^t E \Omega_{J^c}} = \mathcal{M}_j\rndBrk{\sigma^{(j,0)}_{\hat{A}_{1}^{j-1} \hat{B}_{1}^{j-1} \hat{X}_{1}^{j-1} \hat{Y}_{1}^{j-1} T_1^{j-1} I_1^t \tilde{E}^{(j)}_A \tilde{E}^{(j)}_B E \Omega_{J^c}}}.
\end{align}
Using Claim \ref{claim:rho_sigma_approx} we further have that for every $j \in [t]$
\begin{align}
    \frac{1}{2}\norm{\rho_{\hat{A}_1^j \hat{B}_1^j \hat{X}_1^j \hat{Y}_1^j T_1^j I_1^t E \Omega_{J^c}} - \sigma^{(j)}_{\hat{A}_1^j \hat{B}_1^j \hat{X}_1^j \hat{Y}_1^j T_1^j I_1^t E \Omega_{J^c}}}_1 \leq \epsilon
\end{align}
for $\epsilon = O\rndBrk{\frac{\delta^{1/16}}{\alpha^3}}$. We choose the testing maps $\mathcal{T}_j: \hat{X}_j \hat{Y}_j T_j \hat{A}_j \hat{B}_j \rightarrow \hat{X}_j \hat{Y}_j T_j \hat{A}_j \hat{B}_j W_j $ to be the classical channel which outputs the register $W_j$ according to 
\begin{align}
    W_j = \begin{cases}
    V(\hat{X}_j, \hat{Y}_j,\hat{A}_j, \hat{B}_j) & \text{ if } T_j = 1 \\
    \perp & \text{ if } T_j = 0
    \end{cases}
\end{align}
We make the following choices in order to use Theorem \ref{th:approx_EAT_wtest},
\begin{align}
    & A_j \leftarrow \hat{A}_j \hat{B}_j \\
    & B_j \leftarrow \hat{X}_j \hat{Y}_j T_j \\
    & X_j \leftarrow W_j \\
    & E \leftarrow E I_1^t \Omega_{J^c}.
\end{align}
Note that the side information $\hat{X}_j \hat{Y}_j T_j $ is sampled independent of the input state by the channel $\mathcal{M}_j$. Lastly, we will condition on the event $\lnot F$, which is equivalent to the event $\text{freq}(W_1^t)(1) \geq \gamma \omega_{\text{th}}$. \\

We will use Lemma \ref{lemm:3CHSH_anch_entropy} to define an affine min-tradeoff function for the channels $\{\cM_j \}_{j=1}^t$. Fix $j \in [t]$. For an arbitrary state $\nu^{(0)}_{\tilde{E}^{(j)}_A \tilde{E}^{(j)}_B R}$, the state $\nu_{\hat{X}_j \hat{Y}_j T_j \hat{A}_j \hat{B}_j W_j \tilde{R}} = \mathcal{T}_j \circ \cM_j \rndBrk{\nu^{(0)}}$ satisfies 
\begin{align}
    \nu_{T_j} = \gamma \ket{1}\bra{1} + (1-\gamma)\ket{0}\bra{0}.
\end{align}
Further, we can write the output state $\nu$ on register $W_j$ as 
\begin{align}
    \nu_{W_j} = \gamma (1-\omega) \ket{0}\bra{0} + \gamma \omega \ket{1}\bra{1} + (1- \gamma) \ket{\perp}\bra{\perp}
    \label{eq:arb_st_W_dist}
\end{align}
for some $\omega \in [0,1]$. Observe that $\omega$ here is actually the winning probability of the 3CHSH$_{\perp}$ game for the strategy given by the initial state $\nu^{(0)}$ and the measurements that comprise the channel $\cM_j$. Let's define the function $F_{\alpha, \nu}$ piecewise as
\begin{align}
    F_{\alpha, \nu}(x) &:= \begin{cases}
        (1- \alpha)F\rndBrk{g_{\alpha, \nu}\rndBrk{\frac{x}{\gamma}}} & \text{ if } \frac{x}{\gamma} \in \sqBrk{\omega_{\min}, \omega_{\max}} \\
        0 & \text{ else }
        \end{cases}
\end{align}
where $\omega_{\min} := 1- \frac{(1-\alpha)^2 \nu}{4}$, $\omega_{\max} := 1- \frac{2-\sqrt{2}}{4} (1-\alpha)^2 \nu $, and the functions $F$ and $g_{\alpha, \nu}$ are defined in Eq. \ref{eq:3CHSH_bd_phi}. Using Lemma \ref{lemm:3CHSH_anch_entropy} for the state $\nu$, we have
\begin{align}
    H(AB| E XY)_{\nu} \geq F_{\alpha, \nu}\rndBrk{\nu_{W_j}(1)}
\end{align}
We can transform this lower bound into an affine function by using the fact that $F_{\alpha, \nu}$ is convex and that a convex function lies above its slope. We consider the slope at point $\gamma \omega_{\text{th}}$. For $0\leq x \leq \gamma \omega_{\max}$, we have 
\begin{align}
    F_{\alpha, \nu}(x) &\geq F_{\alpha, \nu}(\gamma \omega_{\text{th}}) + F'_{\alpha, \nu}(\gamma \omega_{\text{th}}) \rndBrk{x- \gamma \omega_{\text{th}}} \\
    &=: \bar{F}_{\alpha, \nu} \rndBrk{x}.
\end{align}
We have defined the right-hand side above as the linear function $\bar{F}_{\alpha, \nu}$. Note that Eq. \ref{eq:arb_st_W_dist} implies that irrespective of the input state $\nu^{(0)}$, $\nu_{W_j}(\perp) = 1- \gamma$. So, for any probability distribution $q_{W_k}$ such that $q(\perp) \neq 1 - \gamma$, we have that $\Sigma (q | \cM_j) = \emptyset$ (Eq. \ref{eq:Sigma_set_defn}). This is also true for any distribution $q$ such that $q(1)/\gamma >  \omega_{\max}$, which is the maximum winning probability for the 3CHSH$_\perp$ game. As a result, for the distribution $q_{W_j}$, the function $q \mapsto \bar{F}_{\alpha, \nu}(q(1))$ is a min-tradeoff function for the channels $\{\cM_j \}_{j=1}^t$. \\

One can easily evaluate the derivative of $F_{\alpha, \nu}$:
\begin{align}
    F'_{\alpha, \nu}(\omega_{\text{th}}) = \frac{1}{\nu \gamma (1-\alpha)}F'(g_{\alpha, \nu}(\omega_{\text{th}})) \leq O\rndBrk{\frac{1}{\nu \gamma}}
\end{align}
This gives us that $\norm{\nabla \bar{F}_{\alpha, \nu}}_\infty \leq O\rndBrk{\frac{1}{\nu \gamma}}$, which is required while applying Theorem \ref{th:approx_EAT_wtest}. Finally, applying the approximate EAT to the state $\rho$ as described above, we get that if $\Pr(\lnot F) \geq 2 \mu$\footnote{If $\Pr(\lnot F) < 2 \mu$, then one can show that the secrecy condition for QKD is satisfied for a security parameter greater than $2\mu$ \cite[Section 4.3]{Portmann14}}, then
\begin{align}
    H_{\min}^{\mu' + \epsilon'}(\hat{A}_1^t \hat{B}_1^t |\hat{X}_1^t \hat{Y}_1^t T_1^t I_1^t \Omega_{J^c} E)_{\rho_{|\lnot F}} &\geq t \rndBrk{(1-\alpha)F\rndBrk{g_{\alpha, \nu} (\omega_{\text{th}})}  - O\rndBrk{\frac{\sqrt{\mu}}{\nu \gamma}}} - O(1)
    \label{eq:Hmin_AB_given_E}
  \end{align}
where 
\begin{align}
    & \delta \in (0,1) \\
    & t = \frac{\delta}{\log |\mathcal{A}||\mathcal{B}| + \delta} n \\
    & \epsilon = O\rndBrk{\frac{\delta^{1/16}}{\alpha^3}} \\
    &\mu := \rndBrk{4(4 \sqrt{\epsilon} + \epsilon) \log\frac{|\mathcal{A}| |\mathcal{B}||\mathcal{X}||\mathcal{Y}|||\mathcal{T}|}{\epsilon}}^{1/3} = O\rndBrk{\epsilon^{1/6} \rndBrk{\log \frac{1}{\epsilon}}^{1/3}}\\
    &\mu' := 2\sqrt{\frac{\mu}{\Pr_{\rho}(\lnot F)}} = O\rndBrk{\epsilon^{1/12} \rndBrk{\log \frac{1}{\epsilon}}^{1/6}} 
    \label{eq:Hmin_AB_fin_bd}
\end{align}
and $\epsilon' = \Omega(1) \in (0,1)$ such that $\mu' + \epsilon' <1$. \\

\begin{sloppypar}
In Appendix \ref{sec:suppl_args}, using standard techniques, we derive a bound for the entropy of Alice's raw key with respect to Eve's information, $H_{\min}^{O(\mu')}(\hat{A}_1^t | T_1^t I_1^t \Omega_{1}^n E X_S A_S)_{\rho_{|\lnot F}}$ from the bound above. We use the fact that if the 3CHSH$_\perp$ games are won with a high probability, then Alice's and Bob's answers have a small relative distance. Consequently, Alice's answers alone possess high entropy relative to Eve's information. We account for the information disclosed during the testing procedure, $X_S$ and $A_S$, by applying a straightforward dimension bound. Through these arguments, we prove that:
\begin{align}
    H^{\mu' + 8\epsilon' }_{\min}&(A_J | J T_1^t \Omega_{1}^n E X_S A_S)_{\rho_{|\lnot F}} - \text{leak}_{\text{IR}} \nonumber \\
    &\geq t \rndBrk{(1-\alpha)F\rndBrk{g_{\alpha, \nu} (\omega_{\text{th}})}  - O\rndBrk{\frac{\sqrt{\mu}}{\nu \gamma}} - 2 h(2(\nu + \alpha + \delta_1)) - 2 \log|\mathcal{A}| \gamma} - O(1)
    \label{eq:ultim_bd0}
\end{align}
where $\text{leak}_{\text{IR}}$ is the information reconciliation cost, $\delta_1 \in (0,1)$ is a small parameter and the rest of the parameters are defined and chosen as in Eq. \ref{eq:Hmin_AB_fin_bd}.\\

We can make the lower bound above at least $\succsim t/2$ by choosing $\omega_{\text{th}}$ such that $g_{\alpha, \nu}(\omega_{\text{th}}) \approx 0.84$ (the winning probability of the CHSH games). This choice results in $F(g_{\alpha, \nu}(\omega_{\text{th}})) \geq 3/4$. Note that this winning probability threshold is sufficiently below the maximal winning probability to allow for a robust implementation that accounts for experimental imperfections. We can further choose $\alpha = \gamma = \nu = \delta_1 = 10^{-3}$. With these choices, the combined terms $\alpha F(g_{\alpha, \nu}(\omega_{\text{th}})) + 2h(2(\nu + \alpha + \delta)) + 2 \gamma \log|\mathcal{A}|$ sum together less than $0.1$. By choosing $\delta$ small enough, the term $O\rndBrk{\frac{\sqrt{\mu}}{\nu \gamma}}$ can be made smaller than $0.1$ and the security parameter $\mu'$ can also be made small enough. Once we fix a value of $\delta$, we obtain a fixed small key rate ($\geq \frac{\delta}{2(\log |\mathcal{A}| |\mathcal{B}| + \delta)}$ in this example) for a fixed security parameter given by the corresponding value of $\mu'$. It is important to note that one limitation of our approach is the interdependence of the key rate and the security parameter. We have essentially proven that the protocol is $\tilde{O}(\epsilon_s)$ secure for a rate $\Omega(\epsilon_s^{192})$. Consequently, if one wishes to make the security parameter smaller, they must also reduce the key rate. 
\end{sloppypar}

\section{Conclusion}

In this work, we have developed an alternative approach to proving the security of parallel DIQKD. This approach is based on using results from the work on parallel repetition for anchored games to break the multi-round parallel strategy used by Alice and Bob's quantum devices into multiple approximately single-round strategies. These strategies can then be analysed using an appropriate counterpart of EAT called the unstructured approximate EAT. Our approach yields a more information theoretic and general proof compared to those presented by \cite{Jain20} and \cite{Vidick17}. \\

However, a major drawback of our technique is that it couples the key rate of the protocol with the security parameter. The most important and immediate problem arising from this work is whether this dependence can be broken. We expect that it might be possible to use stronger properties of testing in the unstructured approximate EAT to break this relation. We believe that efforts to enhance the performance of our entropic method to match and potentially surpass that of \cite{Jain20} and \cite{Vidick17} could reveal deeper insights about approximation chains. Furthermore, it would also be interesting to explore whether the results and techniques employed in this paper have implications in the broader context of parallel repetition.\\

At a broader level, whether we can improve the rates of parallel DIQKD to match those of its sequential counterpart still remains open. Given the complexity inherent to general non-local game and their parallel repetitions, it seems unlikely that one can match the performance of sequential DIQKD without relying on specific properties of the underlying games. A promising direction could be to select or engineer these games such that the parallel key rates approach those of sequential DIQKD. For example, certain games like XOR game satisfy strong parallel repetition \cite{Cleve08}, that is, the strategy of playing each game in the parallel repetition independently using the optimal strategy yields the optimal winning probability. It's conceivable that for a similar class of games, Eve's optimal attack in parallel DIQKD might be constrained to be i.i.d. or close to it. Finally, we note that proving the security of parallel device-independent randomness expansion still remains an interesting open problem.

\appendix

\addcontentsline{toc}{section}{APPENDICES}
\section*{APPENDICES}

\section{Single-round results}
\label{sec:app_single_box}

\begin{lemma}
	Suppose Alice and Bob use a strategy, which wins the 3CHSH$_{\perp}$ game with probability $\omega$. Then, their answers $A$ and $B$ for the game satisfy
	\begin{align}
		\Pr[A=B]\geq \omega-  \nu - 2\alpha.
	\end{align}
	\label{lemm:AEqBProb}
\end{lemma}
\begin{proof}
Let $P_{XY}$ represent the probability distribution of the questions for the 3CHSH$_{\perp}$ game. It follows that the winning probability of the 3CHSH$_{\perp}$ game satisfies:
\begin{align*}
   \omega &= \sum_{x ,y} P_{XY}(xy) \sum_{a,b:V(x,y,a,b)=1} P_{AB|XY} (ab|xy) \\
   &= (1-\alpha)^2 (1-\nu) \Pr [A=B | X=0, Y=2] + (1-\alpha)^2 \nu \Pr ({A \oplus B =XY|X,Y \in \{0,1\}} )\\
   & \quad+ (1-(1-\alpha)^2) \\
   & \leq (1-\alpha)^2 (1-\nu) \Pr [A=B | X=0, Y=2] + (1-\alpha)^2 \nu + (1-(1-\alpha)^2) \\
   & \leq \Pr [A=B] + (1-\alpha)^2 \nu + (1-(1-\alpha)^2) \\
   & \leq \Pr [A=B] + \nu + 2\alpha.
\end{align*}
\end{proof}

\begin{lemma}[{\cite[Lemma 5.3]{Friedman20}}]
	Suppose that the quantum strategy for the 2CHSH game starting with $\rho_{E_A E_B E}^{(0)}$ wins the 2CHSH game with probability $\omega \in \lr{[ \frac{3}{4}, \frac{2+\sqrt{2} }{4}}]$. Let $A$ be Alice's answer produced according to the given strategy. Let $\rho_{XAE}$ be the state produced once Alice applies her measurements to $\rho^{(0)}$. Then, for question $x\in \{0,1\}$ we have
	\begin{align}
		H(A| E)_{\rho^{(x)}} \geq F(\omega)
	\end{align}
    where $F(\omega) = 1 - h \rndBrk{\frac{1}{2} + \frac{1}{2} \sqrt{3 - 16\, \omega\left(1 - \omega\right)}}$ and $\rho^{(x)}_{AE}$ are the $A$ and $E$ registers of $\rho_{XAE}$ for $X=x$.
	\label{lemm:2CHSH_entropy}
 \end{lemma}
 \begin{proof}
	This is proved as an intermediate step in the Proof of Lemma 5.3 (Appendix C.1) \cite{Friedman20}.
 \end{proof}

\begin{lemma}
   Suppose that the quantum strategy for the 3CHSH game starting with $\rho_{E_A E_B E}^{(0)}$ wins the 3CHSH game with probability $\omega \in \lr{[ (1-\nu)+\frac{3}{4}\nu, (1-\nu)+\frac{2+\sqrt{2} }{4}\nu}]$. Let $A$ be Alice's answer produced according to the strategy. Then, for the post measurement state $\rho_{XAE}$ we have
   \begin{align*}
	   H(A| E X)_{\rho} \geq F(\hat{\omega}_{\nu})
   \end{align*}
   where $\hat{\omega}_{\nu} := \frac{\omega-(1-\nu)}{\nu}$ and $F$ is as in Lemma \ref{lemm:2CHSH_entropy}
\label{lemm:3CHSH_entropy}
\end{lemma}
\begin{proof}
Let $S=(\rho_{E_A E_B E}^{(0)}, \{A_x \}_{x \in \mathcal{X}}, \{B_y \}_{y \in \mathcal{Y}})$ be the strategy for the 3CHSH mentioned in the lemma hypothesis. Let $S'$ be the strategy for the 2CHSH game, which uses the state $\rho_{E_A E_B E}^{(0)}$, the measurements $\{A_x \}_{x \in \{ 0,1\}}$ as Alice's measurements and the measurements $\{B_y \}_{y \in \{ 0,1\} }$ as Bob's measurements. Let $P_{XY}$ be the distribution of questions in the 2CHSH game, $P_{AB|XY}$ be the conditional probability distribution of the answers in the strategy $S$ and $Q_{AB|XY}$ in the strategy $S'$. Then, by definition of the strategy $S'$, we have for all $x, y \in \{0,1\}$ that 
\begin{align*}
   Q_{AB|X = x, Y = y} = P_{AB|X = x,Y = y}.
\end{align*}
Now, observe that the winning probability $\omega$ of the 3CHSH game can be written as 
\begin{align*}
   \omega &= (1- \nu) P [A=B | X=0, Y =2] + \nu \sum_{x , y \in \{ 0,1\}} P_{XY}(xy) \sum_{a,b: a \oplus b =xy} P_{AB|XY}(ab|xy)\\
   & \leq  (1- \nu) + \nu \sum_{x , y \in \{ 0,1\}} P_{XY}(xy) \sum_{a,b: a \oplus b =xy} P_{AB|XY}(ab|xy) \\
   & =  (1- \nu) + \nu \sum_{x , y \in \{ 0,1\}} P_{XY}(xy) \sum_{a,b: a \oplus b =xy} Q_{AB|XY}(ab|xy) \\
   & =  (1- \nu) + \nu \omega_{S'}
\end{align*}
where $\omega_{S'}$ is the winning probability for the 2CHSH game with the strategy $S'$. The above implies
\begin{align}
	\omega_{S'} \geq \frac{\omega - (1-\nu)}{\nu} = \hat{\omega}_{\nu}.
\end{align}
Note that $\hat{\omega}_{\nu} \geq 3/4$ for the range of $\omega$ in the hypothesis. Using Lemma \ref{lemm:2CHSH_entropy}, we have for $x \in \{ 0,1 \}$
\begin{align*}
   H(A| E)_{\rho_x} \geq 1- h \lr{(\frac{1}{2}+ \frac{1}{2}\sqrt{3- 16\hat{\omega}_{\nu}(1- \hat{\omega}_{\nu})} })
\end{align*}
where $\rho_{AE}^{(x)}$ is the state produced at the end of the strategy $S'$ for question $X=x$. This is the same as the state produced at the end of the strategy $S$ for question $X=x$ because of the way we have defined $S'$. Therefore, for the state produced at the end of $S$, we have
\begin{align*}
   H(A|EX)_{\rho} &= \rndBrk{1-\frac{\nu}{2}}H(A|E)_{\rho_{|X=0}} + \frac{\nu}{2}H(A|E)_{\rho_{|X=1}} \\
   & \geq 1- h \lr{(\frac{1}{2}+ \frac{1}{2}\sqrt{3- 16\hat{\omega}_{\nu}(1- \hat{\omega}_{\nu})} })
\end{align*}
which proves the lemma.
\end{proof}

\begin{proof}[Proof of Lemma \ref{lemm:3CHSH_anch_entropy}]
    Let $\mathcal{X}$ and $\mathcal{Y}$ represent the questions for the 3CHSH game, and $\mathcal{X}_\perp$ and $\mathcal{Y}_\perp$ the questions for the 3CHSH$_\perp$ game. Let $S=(\rho_{E_A E_B E}^{(0)}, \{A_x \}_{x \in \mathcal{X}_{\perp}}, \{B_y \}_{y \in \mathcal{Y}_{\perp}})$ be the strategy which wins the 3CHSH$_{\perp}$ with probability $\omega$. Let $\bar{S}$ be the strategy for the 3CHSH game, which uses the state $\bar{\rho}^{(0)}_{E_A E_B E} = \rho^{(0)}$, the measurements $\{A_x \}_{x \in \mathcal{X}}$ as Alice's measurements and the measurements $\{B_y \}_{y \in \mathcal{Y}}$ as Bob's measurements. Let $P_{XY}$ be the distribution of questions for the 3CHSH$_{\perp}$ game, $Q_{XY}$ be the distribution of questions for the 3CHSH game, $P_{AB|XY}$ be the conditional probability distribution of the answers in the strategy $S$ and $Q_{AB|XY}$ in the strategy $\bar{S}$. Then, by definition of the strategy $\bar{S}$, we have for all $x \in \mathcal{X}$ and $y \in \mathcal{Y}$, 
 \begin{align*}
     Q_{AB|X = x, Y = y} = P_{AB|X = x,Y = y}.
 \end{align*}
 Now, observe that the winning probability of the 3CHSH$_{\perp}$ game can be written as 
 \begin{align*}
    \omega &= \sum_{x \in \mathcal{X}_{\perp}, y \in \mathcal{Y}_{\perp}} P_{XY}(xy) \sum_{a,b: V(ab|xy)=1} P_{AB|XY}(ab|xy)\\
    &= (1- (1-\alpha)^2) + (1-\alpha)^2 \sum_{x \in \mathcal{X}, y \in \mathcal{Y}} Q_{XY}(xy) \sum_{a,b: V(ab|xy)=1} P_{AB|XY}(ab|xy) \\
    &= (1- (1-\alpha)^2) + (1-\alpha)^2 \omega_{\bar{S}}
 \end{align*}
 where $\omega_{\bar{S}}$ is the winning probability for the 3CHSH game under the strategy $\bar{S}$. Thus, we have 
 \begin{align*}
    \omega_{\bar{S}}
    &= 1 - \frac{1 - \omega}{ (1-\alpha)^2}. 
 \end{align*}
 Observe that $\omega_{\bar{S}} \in \lr{[ (1-\nu)+\frac{3}{4}\nu, (1-\nu)+\frac{2+\sqrt{2} }{4} \nu }]$ for the range of $\omega$ in the hypothesis for the theorem. Let $\rho_{AE}^{(x)}$ and $\bar{\rho}_{AE}^{(x)}$ denote the states at the end of the protocol $S$ and $\bar{S}$ when Alice receives the question $X=x$. Then, these two states are equal for $x \neq \perp $. Further, the probability distribution $P_{X|X \neq \perp} = Q_{X}$, which implies that 
 \begin{align*}
    \rho_{XAE|X \neq \perp} &= \sum_{x \in \mathcal{X}} P_{X|X \neq \perp}(x) \puretomixed{x} \otimes \rho_{AE}^{(x)} \\
    &= \sum_{x \in \mathcal{X}} Q_X (x) \puretomixed{x} \otimes \bar{\rho}_{AE}^{(x)} \\
    &= \bar{\rho}_{XAE}.
 \end{align*}
 We can use these to create the entropy bound as follows 
 \begin{align*}
    H(AB | E XY)_{\rho} &\geq H(A | E XY)_{\rho} \\
    &= H(A|EX)_{\rho} \\
    &\geq (1-\alpha) H(A | E X)_{\rho_{|X\neq \perp}} \\
    &= (1- \alpha) H(A | E X)_{\bar{\rho}} \\
    &\geq (1- \alpha) \rndBrk{1- h \rndBrk{\frac{1}{2}+ \frac{1}{2}\sqrt{3- 16 g_{\alpha, \nu}(\omega) (1 - g_{\alpha, \nu}(\omega))}}}
 \end{align*}
 where the first line follows from the fact that $B$ is classical, the second line from the no-signalling property which implies that $Y \leftrightarrow X \leftrightarrow AE$, and in the last line we have used Lemma \ref{lemm:3CHSH_entropy}.
 \end{proof}

 \section{Adapting statements from \cite{Bavarian21} to our setting}
 \label{sec:anch_games_result_carry}

 In the parallel DIQKD setting, Eve distributes the registers $E_A$ and $E_B$ of the state $\psi_{E_A E_B E}$ between Alice and Bob, who then play the $n$ anchored games parallelly using their states. In \cite{Bavarian21}, the situation is almost the same, except there are only two parties, Alice and Bob. We can use the results proven in \cite{Bavarian21} by simply introducing a register for Eve in the state $\psi$ in their setting and tracking it through their proofs. The objective for Alice and Bob is also different in \cite{Bavarian21}. They seek to maximise the winning probability of the parallel repetition game. However, since \cite{Bavarian21} considers an arbitrary strategy for parallel repetition, no modification is required on this account. Lastly, \cite{Bavarian21} also considers conditioning on an event $W_C$, representing the two parties winning a subset $C$ of the rounds. The results only rely on this event being determined by the variable $r_{-i}$, so we can simply replace $W_C$ by the trivial event for our case.\\

 In their proofs, \cite{Bavarian21} uses the fact that the state between Alice and Bob can be assumed to be symmetric, which we cannot necessarily guarantee with three parties. However, it is straightforward to also carry out their proofs without using this assumption. \\
 
 We will go through the statements considered in \cite{Bavarian21} till Section 6 and briefly explain the modifications required to prove them in our setting. The numbering in the following list follows \cite{Bavarian21}. \\
 
 \begin{enumerate}
    \item \textbf{Classical results:} All facts involving only classical variables follow for our setting from the same arguments, since we can consider Eve's register $E$ as being a part of Alice's register $E_A$, and Alice's measurements as $\{ A^{E_A}_{x_1^n} (a_1^n) \otimes \Id^{E} \}$. This reduces our setting to the one in \cite{Bavarian21} for these results.
     \item \textbf{Lemma 4.6}: This is true in our case because only classical variables and their properties are used here.
     \item \textbf{Proposition 4.9}: This is true because the argument traces over all the registers in the states $\Phi$. Thus, we can again view $E$ as being a part of $E_A$ here.
     \item \textbf{Eq. 33-34}: The definitions of $\Xi$ and $\Lambda$ should include the register $E$ as part of $\psi$. Same for the definitions of $\xi$ and $\lambda$.
     \item \textbf{Claim 5.13}: We instead need to prove that 
     \begin{align*}
         & \Ex_I \Ex_{R|W_C} \lr{[   I(Y_i : E_A E)_{\xi_r}  }] =O(\delta) \\
         & \Ex_I \Ex_{R|W_C} \lr{[   I(X_i : E_B E)_{\lambda_r}  }] =O(\delta)
     \end{align*}
     \cite{Bavarian21} proves the first claim and notes that the second one is similar. To prove the above results, we simply need to follow the proof in \cite{Bavarian21}, and note that all instances of $E_A$ in the proof can be replaced with $E_A E$. Also, note that Eq. 38 in \cite{Bavarian21} can be derived without assuming that $\psi$ is symmetric. To prove the second claim, we would similarly consider $E_B E$ together.
     \item \textbf{Claim 5.14}: We instead need to prove that 
     \begin{align*}
         & \Ex_I \Ex_{R_{-i}|W_C} \Ex_{XY} \Vert \xi^{E_A E}_{r_{-i},x,y} - \xi^{E_A E}_{r_{-i}, x,\perp} \Vert_1^2  =O(\delta^{1/2}/\alpha^4)\\
         & \Ex_I \Ex_{R_{-i}|W_C} \Ex_{XY} \Vert \lambda^{E_B E}_{r_{-i},x,y} - \lambda^{E_B E}_{r_{-i},\perp, y} \Vert_1^2 =O(\delta^{1/2}/\alpha^4).
     \end{align*}
     Once again this can be done by replacing $E_A$ by $E_A E$ in the proof for the first claim and $E_B$ by $E_B E$ for the second claim. The proof also uses relations between several classical variables, which can be handled using the argument for classical variables mentioned above. 
     \item \textbf{Proof of Lemma 5.12}:
         \begin{enumerate}
             \item \textbf{Claim 5.15}: We need to prove that $\ket{\tilde{\Phi}_{r_i, \perp/x, y}}^{E_A E_B E}$ is a purification of the state $\xi^{E_A E}_{r_i, x, y}$ and that $\ket{\tilde{\Phi}_{r_i, \perp, y}}^{E_A E_B E}$ is a purification of the state $\xi^{E_A E}_{r_i, \perp, y}$. Once again, this can be done by noting that we don't necessarily need to use the fact that the state is symmetric for the proof. Following the same procedure as \cite{Bavarian21}, we can derive
              \begin{align*}
                  \xi^{E_A E}_{r_i, x, y} &= \gamma^{-2}_{r_i, \perp/x, y} A_\omega (a_C)^{1/2} \tr_{E_B} \rndBrk{ B_{\omega_{-i}, y} (b_C) \Psi  } A_\omega (a_C)^{1/2} \\
                  &= \tilde{\Phi}_{r_i, \perp/x, y}^{E_A E}.
              \end{align*}
              The rest of the proof remains the same.
              \item\textbf{Claim 5.16}: Also needs to be modified as Claim 5.15 has been.
         \end{enumerate}
         The rest of the steps in the proof rely on using claims, which were proven before. Since, we have already proven them above, the steps follow in our case as well.
         \item \textbf{Lemma 5.17}: Since the definition of $\gamma_{r_{-i}, s, y}$ traces over all the quantum registers, we can simply view our setting as an instance of the anchored games setting by letting Alice keep the $E$ register. Note that this proof does not use any property about the event $W_C$ beyond the fact that it is determined by $r_{-i}$. 
         \item \textbf{Proof of Proposition 5.1}: All the lemmas and facts required for the proof of Proposition 5.1 have been shown to be valid in our setting. One can now simply follow the proof given in \cite{Bavarian21} to prove the proposition.
 \end{enumerate}

\section{Supplementary arguments for security proof}
\label{sec:suppl_args}

In this section, we continue the argument from the bound in Eq. \ref{eq:Hmin_AB_given_E} and show that it can be transformed into a lower bound for $H_{\min}^{O(\mu')}(\hat{A}_1^t | T_1^t I_1^t \Omega_{1}^n E X_S A_S)_{\rho_{|\lnot F}}$. We also upper bound the information leakage during the information reconciliation phase.  

\subsection{Removing $\hat{B}_1^t$ from the smooth min-entropy bound}

We begin by removing the $\hat{B}_1^t$ registers from the entropy in $H_{\min}^{\mu' + \epsilon'}(\hat{A}_1^t \hat{B}_1^t |\hat{X}_1^t \hat{Y}_1^t T_1^t I_1^t \Omega_{J^c} E)_{\rho_{|\lnot F}}$. For this, it is  sufficient to prove that the entropy 
\begin{align}
    H^{\epsilon'}_{\max}(\hat{B}_1^t | \hat{A}_1^t \hat{X}_1^t \hat{Y}_1^t T_1^t)_{\rho_{|\lnot F}}
\end{align}
is small. Intuitively, this should be true because the average winning probability is at least $\omega_{\text{th}} \geq 1 - \nu$, which implies using Lemma \ref{lemm:AEqBProb} that 
\begin{align}
    \Pr[A=B]\geq 1 - 2\nu - 2\alpha.
\end{align}
So, we should be able to prove that 
\begin{align}
    H^{\epsilon'}_{\max}(\hat{B}_1^t | \hat{A}_1^t \hat{X}_1^t \hat{Y}_1^t T_1^t)_{\rho_{|\lnot F}}\lesssim t\ h(2(\nu + \alpha)).
\end{align}
To prove this, let $\bar{J} := \curlyBrk{j \in J : {X}_j, {Y}_j = (0, 2)}$ for the state $\rho$ (unconditioned). Define the events, 
\begin{align}
    E_1 &:= \sqBrk{\frac{|S|}{t} \geq \gamma - \delta_1} \\
    E_2 &:= \sqBrk{\frac{|\bar{J}|}{t} \geq (1-\alpha)^2(1-\nu) - \delta_1}\\
    E_3 &:= \sqBrk{\frac{1}{t}\sum_{i \in J} V({X}_i, {Y}_i, {A}_i, {B}_i) \geq \omega_{\text{th}} - \delta_1}
\end{align}
for some small parameter $\delta_1 \in (0,1)$. Using the Chernoff-Hoeffding bound, we have that 
\begin{align}
    \Pr[E_1^c] &\leq e^{-\Omega(\delta_1^2 t)} \\
    \Pr[E_2^c] &\leq e^{-\Omega(\delta_1^2 t)}.
\end{align}
Following \cite{Vidick17} (which uses \cite[Lemma 6]{Tomamichel17}), we also have that
\begin{align}
    \Pr[\lnot F \wedge E_3^c] \leq e^{-\Omega(\delta_1^2 \gamma t)}.
\end{align}
Therefore, we have that conditioned on the event $\lnot F$, $E_1\wedge E_2 \wedge E_3$ hold except with probability $\frac{e^{-\Omega(\delta_1^2 \gamma t)}}{\Pr_{\rho}(\lnot F)}$, i.e., 
\begin{align}
    \frac{1}{2}\norm{\rho_{T_1^t \hat{X}_1^t \hat{Y}_1^t \hat{A}_1^t \hat{B}_1^t | \lnot F} -  \rho_{T_1^t \hat{X}_1^t \hat{Y}_1^t \hat{A}_1^t \hat{B}_1^t, E_1 \wedge E_2 \wedge E_3| \lnot F}}_1 \leq \frac{e^{-\Omega(\delta_1^2 \gamma t)}}{\Pr_{\rho}(\lnot F)}. 
    \label{eq:Hmax_bd_dist}
\end{align}
Let $\delta(x, y)$ be the Kronecker delta function, which is $1$ if $x = y$ and $0$ otherwise. Let $e$ be the relative error between $\hat{A}_1^t$ and $\hat{B}_1^t$. If the events $E_1 \wedge E_2 \wedge E_3$ are true, then we have
\begin{align*}
    \omega_{\text{th}} - \delta_1 &\leq \frac{1}{t} \sum_{i \in J} V(X_i, Y_i, A_i, B_i) \\
    &= \frac{1}{t} \sum_{i \in \bar{J}} V(X_i, Y_i, A_i, B_i) + \frac{1}{t} \sum_{i \in J \setminus \bar{J}} V(X_i, Y_i, A_i, B_i)\\
    &\leq \frac{1}{t} \sum_{i \in \bar{J}} \delta(A_i, B_i) + \frac{1}{t} \sum_{i \in J \setminus \bar{J}} (1+\delta(A_i, B_i))\\
    &\leq 1-e + \frac{t-|\bar{J}|}{t} \\
    &\leq 1-e + 1- (1-\alpha)^2(1-\nu) + \delta_1 \\
    &\leq 1- e + \nu + 2\alpha + \delta_1
\end{align*}
which implies that 
\begin{align}
    e \leq 1 - \omega_{\text{th}} + \nu + 2\alpha + 2 \delta_1.
\end{align}
Further, since $\omega_{\text{th}} \geq 1- \nu$, we have that 
\begin{align}
    e \leq 2(\nu + \alpha + \delta_1). 
\end{align}
This enables us to bound the max-entropy for the state $\rho_{E_1 \wedge E_2 \wedge E_3| \lnot F}$:
\begin{align}
    H_{\max}( \hat{B}_1^t | \hat{A}_1^t )_{\rho_{E_1 \wedge E_2 \wedge E_3| \lnot F}} \leq t \cdot h(2(\nu + \alpha + \delta_1)).
\end{align}
Combining with Eq. \ref{eq:Hmax_bd_dist} shows that 
\begin{align}
    H^{\epsilon''}_{\max}(\hat{B}_1^t | \hat{A}_1^t)_{\rho_{|\lnot F}} \leq t \cdot h(2(\nu + \alpha + \delta_1))
    \label{eq:Hmax_B_given_A_bd}
\end{align}
for $\epsilon'' := \frac{e^{-\Omega(\delta_1^2 \gamma t)}}{\sqrt{\Pr(\lnot F)}} = e^{-\Omega(n)}$. \\

We can use this to bound the entropy of Alice's raw key alone by using the chain rule in \cite[Theorem 15]{Vitanov13} and Eq. \ref{eq:Hmin_AB_given_E}:
\begin{align*}
    H_{\min}^{\mu' + 5\epsilon'}&(\hat{A}_1^t |\hat{X}_1^t \hat{Y}_1^t T_1^t I_1^t \Omega_{J^c} E)_{\rho_{|\lnot F}} \\
    &\geq H_{\min}^{\mu' + \epsilon'}(\hat{A}_1^t \hat{B}_1^t |\hat{X}_1^t \hat{Y}_1^t T_1^t I_1^t \Omega_{J^c} E)_{\rho_{|\lnot F}} - H^{\epsilon'}_{\max}(\hat{B}_1^t | \hat{A}_1^t \hat{X}_1^t \hat{Y}_1^t T_1^t I_1^t \Omega_{J^c} E)_{\rho_{|\lnot F}} - O\rndBrk{\log\frac{1}{\epsilon'}}\\
    &\geq H_{\min}^{\mu' + \epsilon'}(\hat{A}_1^t \hat{B}_1^t |\hat{X}_1^t \hat{Y}_1^t T_1^t I_1^t \Omega_{J^c} E)_{\rho_{|\lnot F}} - H^{\epsilon'}_{\max}(\hat{B}_1^t | \hat{A}_1^t)_{\rho_{|\lnot F}} - O\rndBrk{\log\frac{1}{\epsilon'}}\\
    & \geq t \rndBrk{(1-\alpha)F_{\alpha, \nu} (\omega_{\text{th}})  - O\rndBrk{\frac{\sqrt{\mu}}{\nu \gamma}} - h(2(\nu + \alpha + \delta_1))} - O(1).
    \numberthis
    \label{eq:Hmin_A_given_E_bd}
\end{align*}
We chose the smoothing of the max-entropy above to be $\epsilon'$ as well for simplicity. Since, $\epsilon'$ is a constant greater than $0$ and $\epsilon'' = e^{-\Omega(n)}$, it is valid to use the bound in Eq. \ref{eq:Hmax_B_given_A_bd} for sufficiently large $n$.

\subsection{Adding $\Omega_J$ to the conditioning register}

\begin{lemma}
    \label{lemm:cond_Markov_ch}
    For $\rho_{\Omega X A E}$ a classical ($\Omega X A$)-quantum ($E$) state which satisfies the Markov chain $ \Omega \leftrightarrow X \leftrightarrow A E$ and an event $F$ determined by $X$ and $A$, i.e., $F \subseteq \mathcal{X} \times \mathcal{A}$, the conditional state $\rho_{\Omega X A E | F}$ also satisfies $ \Omega \leftrightarrow X \leftrightarrow A E$. 
\end{lemma}
\begin{proof}
    Since $\rho_{\Omega X A E}$ is a classical ($\Omega X A$)-quantum ($E$) state and satisfies the Markov chain $ \Omega \leftrightarrow X \leftrightarrow A E$, $\rho$ is of the form
    \begin{align}
        \rho_{\Omega X A E} = \sum_{x} \rho(x) \puretomixed{x} \otimes \rndBrk{\sum_{\omega} \rho(\omega | x) \puretomixed{\omega}} \otimes \rndBrk{\sum_{a} \rho(a|x) \puretomixed{a} \otimes \rho_{E|a,x}}.
    \end{align}
    Let $\rho(F) = \sum_{x, a \in F} \rho(x) \rho(a | x)$ be the probability of the event $F$. The conditional state $\rho_{|F}$ can be written as
    \begin{align*}
        \rho_{\Omega X A E | F} &= \frac{1}{\rho(F)}\rho_{\Omega X A E \wedge F} \\
        &= \frac{1}{\rho(F)} \sum_{x, a \in F} \rho(x) \rho(a|x) \puretomixed{x, a} \otimes \rho_{E|a,x} \otimes \rndBrk{\sum_{\omega} \rho(\omega | x) \puretomixed{\omega}}\\
        &= \sum_{x} \rho(x|F) \puretomixed{x} \otimes \rndBrk{\sum_{a} \rho(a|x, F) \puretomixed{a} \otimes \rho_{E|a,x}} \otimes \rndBrk{\sum_{\omega} \rho(\omega | x) \puretomixed{\omega}}
    \end{align*}
    which clearly satisfies the Markov chain $ \Omega \leftrightarrow X \leftrightarrow A E$. 
\end{proof}

Note that using Lemma \ref{lemm:MarkovChainOmegaXA}, we have that the state $\rho$ (unconditioned) satisfies the Markov chain
\begin{align}
    \Omega_J \leftrightarrow J X_J Y_J \leftrightarrow A_1^n B_1^n \Omega_{J^c} E T_1^t. 
\end{align}
More precisely, according to Lemma \ref{lemm:MarkovChainOmegaXA} the above is true for every fixed $J$, and the above follows from using simple facts about Markov chains. Further, the event $\lnot F$ is defined using the variables $ J X_J Y_J A_J B_J T_1^t$, so using Lemma \ref{lemm:cond_Markov_ch}, we have that the state $\rho_{|\lnot F}$ also satisfies the above Markov chain condition. Note that tracing over $A_{J^c} B_{J^c}$ does not alter this.\\

Therefore, there exists a channel $\Phi: J X_J Y_J \rightarrow J X_J Y_J \Omega_J$ such that 
\begin{align}
    \Phi\rndBrk{\rho_{|\lnot F}^{A_J B_J X_J Y_J J \Omega_{J^c} T_1^t E}} = \rho_{|\lnot F}^{A_J B_J X_J Y_J J \Omega_{1}^n T_1^t E}.
\end{align}
This implies that 
\begin{align}
    H_{\min}^{\mu' + 5\epsilon'}&(A_J |J X_J Y_J T_1^t \Omega_{1}^n E)_{\rho_{|\lnot F}} = H_{\min}^{\mu' + 5\epsilon'}(A_J |J X_J Y_J T_1^t \Omega_{J^c} E)_{\rho_{|\lnot F}}.
    \label{eq:added_omega_J}
\end{align}
We have already bounded the right-hand side above in Eq. \ref{eq:Hmin_A_given_E_bd}. 

\subsection{Accounting for information leakage during testing}

Finally, we also need to consider the entropy loss due to Alice transmitting $A_S$ in plaintext to Bob. Using the chain rule \cite[Theorem 14]{Vitanov13} we have that
\begin{align*}
    H&^{\mu' + 8\epsilon'}_{\min}(A_J | X_J Y_J T_1^t J \Omega_{1}^n E A_S)_{\rho_{|\lnot F}} \\
    &\geq H^{\mu' + 5\epsilon'}_{\min}(A_J | X_J Y_J T_1^t J \Omega_{1}^n E)_{\rho_{|\lnot F}} - H^{\epsilon'}_{\max}(A_S | X_J Y_J T_1^t J \Omega_{1}^n E)_{\rho_{|\lnot F}} - O\rndBrk{\log \frac{1}{\epsilon'}}. \numberthis
    \label{eq:int_removing_As}
\end{align*}
We can show that the max-entropy above is small for sufficiently large $n$. Since, $T_i$ are chosen in an i.i.d fashion with probability $P(T_i = 1) = \gamma$, we have that with probability at least $1 - e^{-\Omega(\gamma^2 t)}$, the number of $T_i$ that are $1$ is at most $2\gamma t$. Let's call this event $Q$. We then have
\begin{align*}
    \rho_{J X_J Y_J A_J B_J T_1^t} &= \rho_{J X_J Y_J A_J B_J} \otimes \rho_{T_1^t} \\
    &\approx_{e^{-\Omega(\gamma^2 t)}} \rho_{J X_J Y_J A_J B_J} \otimes \rho_{T_1^t \wedge Q}.
\end{align*}
Let $\eta_{J X_J Y_J A_J B_J T_1^t} := \rho_{J X_J Y_J A_J B_J} \otimes \rho_{T_1^t \wedge Q}$. Using \cite[Lemma G.1]{Marwah24_approx_ch}, we have that 
\begin{align}
    \frac{1}{2}\norm{\rho_{X_J Y_J A_J B_J T_1^t|\lnot F} - \eta_{X_J Y_J A_J B_J T_1^t | \lnot F}}_1 \leq \frac{e^{-\Omega(\gamma^2 t)}}{\Pr_{\rho}(\lnot F)}
\end{align}
and hence $P(\rho_{X_J Y_J A_J B_J T_1^t|\lnot F}, \eta_{X_J Y_J A_J B_J T_1^t|\lnot F}) \leq e^{-\Omega(\gamma^2 t)}/ \sqrt{\Pr_{\rho}(\lnot F)}$. Note that for a fixed value of $T_1^t = \tau_1^t$, the state $\eta_{A_S | \tau_1^t}$ has a support of size at most $|\mathcal{A}|^{2\gamma t}$. Therefore, for $\epsilon' = \Omega(1) \geq e^{-\Omega(\gamma^2 t)}/ \sqrt{\Pr_{\rho}(\lnot F)} = e^{-\Omega(n)}$, we have that
\begin{align*}
    H^{\epsilon'}_{\max}(A_S | J X_J Y_J T_1^t \Omega_{1}^n E)_{\rho_{|\lnot F}} \leq 2\gamma t \log|\mathcal{A}|. 
\end{align*}
Plugging this in Eq. \ref{eq:int_removing_As} and using Eq. \ref{eq:Hmin_A_given_E_bd} and \ref{eq:added_omega_J}, we get the bound
\begin{align*}
    H&^{\mu' + 8\epsilon' }_{\min}(A_J | J X_J Y_J T_1^t \Omega_{1}^n E A_S)_{\rho_{|\lnot F}}\\
    & \geq t \rndBrk{(1-\alpha)F_{\alpha, \nu} (\omega_{\text{th}})  - O\rndBrk{\frac{\sqrt{\mu}}{\nu \gamma}} - h(2(\nu + \alpha + \delta_1)) - 2 \log|\mathcal{A}| \gamma} - O(1) 
    \numberthis
    \label{eq:penultim_bd}
\end{align*}
Note that 
\begin{align}
    H&^{\mu' + 8\epsilon' }_{\min}(A_J | J T_1^t \Omega_{1}^n E X_S A_S)_{\rho_{|\lnot F}} \geq H^{\mu' + 8\epsilon'}_{\min}(A_J | J X_J Y_J T_1^t \Omega_{1}^n E A_S)_{\rho_{|\lnot F}}.
\end{align}
Using the bound in Eq. \ref{eq:penultim_bd} in the equation above, we have a linear lower bound for the smooth min-entropy of Alice's raw key ($A_J$) with respect to Eve's state ($J T_1^t \Omega_{1}^n E X_S A_S$):
\begin{align}
    H^{\mu' + 8\epsilon' }_{\min}&(A_J | J T_1^t \Omega_{1}^n E X_S A_S)_{\rho_{|\lnot F}} \nonumber\\
    &\geq t \rndBrk{(1-\alpha)F_{\alpha, \nu} (\omega_{\text{th}})  - O\rndBrk{\frac{\sqrt{\mu}}{\nu \gamma}} - h(2(\nu + \alpha + \delta_1)) - 2 \log|\mathcal{A}| \gamma} - O(1)
    \label{eq:ultim_bd1}
\end{align}

\subsection{Information reconciliation cost}

In the one-shot setting, the information reconciliation cost, denoted as $\text{leak}_{\text{IR}}$, is given by $H_{\max}^{\epsilon''}(\hat{A}_1^t | \hat{B}_1^t I_1^t)_{\rho_{|\lnot F}}$ up to a constant (see \cite[Section 4.2.2]{Friedman20} for additional details). \\

We have already bound the entropy $H_{\max}^{\epsilon''}(\hat{B}_1^t | \hat{A}_1^t I_1^t)_{\rho_{|\lnot F}}$ in Eq. \ref{eq:Hmax_B_given_A_bd}. The same argument and bound can also be used for the entropy above. So, we have 
\begin{align}
    \text{leak}_{\text{IR}} \leq H_{\max}^{\epsilon''}(\hat{A}_1^t | \hat{B}_1^t I_1^t)_{\rho_{|\lnot F}} + O(1) &\leq t \cdot h(2(\nu + \alpha + \delta_1)) + O(1)
\end{align}
for $\epsilon'' := \frac{e^{-\Omega(\delta_1^2 \gamma t)}}{\sqrt{\Pr(\lnot F)}} = e^{-\Omega(n)}$. \\

\subsection{Key length}

Up to a constant factor the key length is given by 
\begin{align}
    H^{\mu' + 8\epsilon' }_{\min}&(A_J | J T_1^t \Omega_{1}^n E X_S A_S)_{\rho_{|\lnot F}} - \text{leak}_{\text{IR}} \nonumber\\
    &\geq t \rndBrk{(1-\alpha)F_{\alpha, \nu} (\omega_{\text{th}})  - O\rndBrk{\frac{\sqrt{\mu}}{\nu \gamma}} - 2 h(2(\nu + \alpha + \delta_1)) - 2 \log|\mathcal{A}| \gamma} - O(1)
    \label{eq:ultim_bd}
\end{align}
where $\alpha, \nu \in (0,0.1)$ are parameters for the 3CHSH$_\perp$ game, and the rest of the parameters are chosen as: 
\begin{align}
    & \delta \in (0,1) \\
    & t = \frac{\delta}{\log |\mathcal{A}||\mathcal{B}| + \delta} n \\
    & \epsilon = O\rndBrk{\frac{\delta^{1/16}}{\alpha^3}} \\
    &\mu = O\rndBrk{\epsilon^{1/6} \rndBrk{\log \frac{1}{\epsilon}}^{1/3}}\\
    &\mu' := 2\sqrt{\frac{\mu}{\Pr_{\rho}(\lnot F)}} = O\rndBrk{\epsilon^{1/12} \rndBrk{\log \frac{1}{\epsilon}}^{1/6}} \\
    &\delta_1 \in (0,1)
\end{align}
and $\epsilon' = \Omega(1) \in (0,1)$ such that $\mu' + \epsilon' <1$. We have assumed here that $\Pr(\lnot F) \geq 2 \mu$.

\bibliographystyle{alpha}
\bibliography{bib}

\end{document}